\documentclass[11pt]{article}%
\usepackage{amsmath}
\usepackage{amsfonts}
\usepackage{amssymb}
\usepackage{graphicx}
\usepackage{fullpage}
\usepackage{hyperref}%
\providecommand{\U}[1]{\protect\rule{.1in}{.1in}}
\newtheorem{theorem}{Theorem}

\newtheorem{corollary}[theorem]{Corollary}

\newtheorem{lemma}[theorem]{Lemma}

\newtheorem{problem}[theorem]{Problem}
\newtheorem{proposition}[theorem]{Proposition}

\newenvironment{proof}[1][Proof]{\noindent\textbf{#1.} }{\ \rule{0.5em}{0.5em}}
\begin{document}

\title{Shadow Tomography of Quantum States\thanks{An earlier version of this paper
appeared in the Proceedings of the 2018 ACM Symposium on Theory of Computing
(STOC), pages 325-338}}
\author{Scott Aaronson\thanks{University of Texas at Austin. \ Email:
aaronson@cs.utexas.edu. \ Supported by a Vannevar Bush Fellowship from the US
Department of Defense, a Simons Investigator Award, and the Simons
\textquotedblleft It from Qubit\textquotedblright\ collaboration.}}
\date{}
\maketitle

\begin{abstract}
We introduce the problem of \textit{shadow tomography}: given an unknown
$D$-dimensional quantum mixed state $\rho$, as well as known two-outcome
measurements $E_{1},\ldots,E_{M}$, estimate the probability that $E_{i}%
$\ accepts $\rho$, to within additive error $\varepsilon$, for each of the $M$
measurements. \ How many copies of $\rho$\ are needed to achieve this, with
high probability? \ Surprisingly, we give a procedure that solves the problem
by measuring only $\widetilde{O}\left(  \varepsilon^{-4}\cdot\log^{4}%
M\cdot\log D\right)  $\ copies. \ This means, for example, that we can learn
the behavior of an arbitrary $n$-qubit state, on \textit{all}
accepting/rejecting circuits of some fixed polynomial size, by measuring only
$n^{O\left(  1\right)  }$\ copies of the state. \ This resolves an open
problem of the author, which arose from his work on private-key quantum money
schemes, but which also has applications to quantum copy-protected software,
quantum advice, and quantum one-way communication. \ Recently, building on
this work, Brand\~{a}o et al.\ have given a different approach to shadow
tomography using semidefinite programming, which achieves a savings in
computation time.

\end{abstract}

\section{Introduction\label{INTRO}}

One of the most striking features of quantum mechanics is the
\textit{destructive nature of measurement}.\ \ Given a single copy of a
quantum state $\rho$, which is otherwise unknown to us, no amount of
cleverness will ever let us recover a classical description of $\rho$, even
approximately, by measuring $\rho$. \ Of course, the destructive nature of
measurement is what opens up many of the cryptographic possibilities of
quantum information, including quantum key distribution and quantum money.

In general, the task of recovering a description of a $D$-dimensional quantum
mixed state $\rho$, given many copies of $\rho$, is called \textit{quantum
state tomography}. \ This task can be shown for information-theoretic
reasons\ to require $\Omega\left(  D^{2}\right)  $\ copies of $\rho$,\ while a
recent breakthrough of O'Donnell and Wright \cite{owright}\ and Haah et
al.\ \cite{hhjwy}\ showed that $O\left(  D^{2}\right)  $\ copies also
suffice.\footnote{Here and throughout this paper, the notations
$\widetilde{\Omega}$\ and $\widetilde{O}$\ mean we suppress polylogarithmic
factors.} \ Unfortunately, this number can be astronomically infeasible:
recall that, if $\rho$ is a state of $n$ entangled qubits, then $D=2^{n}$.
\ No wonder that the world record, for full\footnote{By \textquotedblleft
full,\textquotedblright\ we mean that the procedure could have recovered the
state $\rho$\ regardless of what it was, rather than requiring an assumption
like, e.g., that $\rho$\ has a small matrix product state description.}
quantum state tomography, is $10$-qubit states, for which millions of
measurements were needed \cite{songetal}.

Besides the practical issue, this state of affairs could be viewed as an
epistemic problem for quantum mechanics itself. \ If learning a full
description of an $n$-qubit state $\rho$\ requires measuring $\exp\left(
n\right)  $ copies of $\rho$, then should we even say that the full
description is \textquotedblleft there\textquotedblright\ at all, in a single
copy of $\rho$?\noindent

Naturally, we could ask the same question about a classical probability
distribution $\mathcal{D}$\ over $n$-bit strings. \ In that case, the
exponentiality seems to militate toward the view that no, the vector of
$2^{n}$\ probabilities is \textit{not} \textquotedblleft out
there\textquotedblright\ in the world, but is only \textquotedblleft in our
heads,\textquotedblright\ while what's \textquotedblleft out
there\textquotedblright\ are just the actual $n$-bit samples from
$\mathcal{D}$, along with whatever physical process generated the samples.
\ Quantum mechanics is different, though, because $2^{n}$\ amplitudes can
interfere with each other: an observable effect that seems manifestly
\textit{not} just in our heads! \ Interference forces us to ask the question anew.

Partly inspired by these thoughts, a long line of research has sought to show
that, once we impose some reasonable operational restrictions on what a
quantum state will be used for, an $n$-qubit\ state $\rho$\ actually contains
\textquotedblleft much less information than meets the eye\textquotedblright:
more like $n$ or $n^{O\left(  1\right)  }$\ classical bits than like $2^{n}$
bits. \ Perhaps the \textquotedblleft original\textquotedblright\ result along
these lines was Holevo's Theorem \cite{holevo}, which says that by sending an
$n$-qubit state, Alice can communicate at most $n$ classical bits to Bob (or
$2n$, if Alice and Bob have pre-shared entanglement). \ Subsequently, the
random access code lower bound of Ambainis, Nayak, Ta-Shma, and Vazirani
\cite{antv}\ showed that this is still true, even if Bob wants to learn just a
single one of Alice's bits.

Since 2004, a series of results by the author and others has carried the basic
conclusion further. \ Very briefly, these results have included the
postselected learning theorem \cite{aar:adv}; the Quantum Occam's Razor
Theorem \cite{aar:learn}; the \textquotedblleft
de-Merlinization\textquotedblright\ of quantum protocols \cite{aar:qmaqpoly};
a full characterization of quantum advice \cite{adrucker}; and a recent online
learning theorem for quantum states \cite{achn}. \ We'll apply tools from
several of those results in this paper, and will discuss the results later in
the introduction where it's relevant. \ In any case, though, none of the
previous results directly addressed the question:
\textit{information-theoretically, how much can be learned about an }%
$n$\textit{-qubit state }$\rho$\textit{\ by measuring only }$n^{O\left(
1\right)  }$\textit{ copies of }$\rho$\textit{?}

\subsection{Our Result\label{RESULT}}

Motivated by the above question, this paper studies a basic new task that we
call \textit{shadow tomography}, and define as follows.

\begin{problem}
[Shadow Tomography]\label{theprob}Given an unknown $D$-dimensional quantum
mixed state $\rho$, as well as known $2$-outcome measurements $E_{1}%
,\ldots,E_{M}$, each of which accepts\ $\rho$\ with probability
$\operatorname{Tr}\left(  E_{i}\rho\right)  $ and rejects\ $\rho$\ with
probability $1-\operatorname{Tr}\left(  E_{i}\rho\right)  $,\ output numbers
$b_{1},\ldots,b_{M}\in\left[  0,1\right]  $\ such that $\left\vert
b_{i}-\operatorname{Tr}\left(  E_{i}\rho\right)  \right\vert \leq\varepsilon
$\ for all $i$, with success probability at least $1-\delta$. \ Do this via a
measurement of $\rho^{\otimes k}$, where $k=k\left(  D,M,\varepsilon
,\delta\right)  $ is as small as possible.
\end{problem}

The name \textquotedblleft shadow tomography\textquotedblright\ was suggested
to us by Steve Flammia, and refers to the fact that we aim to recover, not the
full density matrix of $\rho$, but only the \textquotedblleft
shadow\textquotedblright\ that $\rho$\ casts on the measurements $E_{1}%
,\ldots,E_{M}$.

Observe, for a start, that shadow tomography is easy to achieve using
$k=O\left(  D^{2}/\varepsilon^{2}\right)  $\ copies of the state $\rho$, by
just ignoring the $E_{i}$'s and doing full tomography on $\rho$, using the
recent protocols of O'Donnell and Wright \cite{owright}\ or Haah et
al.\ \cite{hhjwy}. \ At a different extreme of parameters, shadow tomography
is \textit{also} easy using $k=\widetilde{O}\left(  M/\varepsilon^{2}\right)
$\ copies of $\rho$, by just applying each measurement $E_{i}$\ to separate
copies of $\rho$.

At a mini-course taught in February 2016 (see \cite[Section 8.3.1]{aarbados}),
the author discussed shadow tomography---though without calling it that---and
posed the question, \textit{what happens if }$D$\textit{ and }$M$\textit{ are
both exponentially large?} \ Is it conceivable that, even then, the $M$
expectation values $\operatorname{Tr}\left(  E_{1}\rho\right)  ,\ldots
,\operatorname{Tr}\left(  E_{M}\rho\right)  $\ could all be approximated using
only, say, $\operatorname*{poly}\left(  \log D,\log M\right)  $\ copies of the
state $\rho$? \ The author didn't venture to guess an answer; other
researchers' opinions were also divided.

The main result of this paper is to settle the question affirmatively.

\begin{theorem}
[Shadow Tomography Theorem]\label{main}Problem \ref{theprob}\ (Shadow
Tomography) is solvable using only%
\[
k=\widetilde{O}\left(  \frac{\log1/\delta}{\varepsilon^{4}}\cdot\log^{4}%
M\cdot\log D\right)
\]
copies of the state $\rho$,\ where the $\widetilde{O}$\ hides a
$\operatorname*{poly}\left(  \log\log M,\log\log D,\log\frac{1}{\varepsilon
}\right)  $ factor.\footnote{In an earlier version of this paper, the
dependence on $\varepsilon$\ was $1/\varepsilon^{5}$. \ The improvement to
$1/\varepsilon^{4}$ comes from using the recent online learning algorithm of
Aaronson et al. \cite{achn}.} \ The procedure is fully explicit.
\end{theorem}

In Section \ref{TECHNIQUES}, we'll give an overview of the proof of this
theorem. \ In Section \ref{MOTIV}, we'll discuss the motivation,\ and give
applications to quantum money, quantum copy-protected software, quantum
advice, and quantum one-way communication. \ For now, let's make some initial
comments about the theorem itself: why it's nontrivial, why it's consistent
with other results, etc.

The key point is that Theorem \ref{main}\ lets us learn the behavior of a
state\ of exponential dimension, with respect to exponentially many different
observables, using only polynomially many copies of the state. \ To achieve
this requires measuring the copies in an extremely careful way, to avoid
destroying them as we proceed.

Naturally, to implement the required measurement on $\rho^{\otimes k}$\ could,
in the worst case, require a quantum circuit of size polynomial in both $M$
and $D$. \ (Note that the input---i.e., the list of measurement operators
$E_{i}$---already involves $\Theta\left(  MD^{2}\right)  $\ complex numbers.)
\ It's interesting to study how much we can improve the \textit{computational}
complexity of shadow tomography, with or without additional assumptions on the
state $\rho$\ and measurements $E_{i}$. \ In Sections \ref{RELATED} and
\ref{OPEN}, we'll say more about this question, and about recent work by
Brand\~{a}o et al.\ \cite{bkllsw}, which builds on our work to address it.
\ In this paper, though, our main focus is on the information-theoretic
aspect, of how many copies of $\rho$\ are needed.

The best \textit{lower} bound that we know on the number of copies is
$\Omega\left(  \frac{\min\left\{  D^{2},\log M\right\}  }{\varepsilon^{2}%
}\right)  $.\footnote{In an earlier version of this paper, we proved only a
weaker lower bound, namely $\Omega\left(  \frac{\log M}{\varepsilon^{2}%
}\right)  $\ assuming $D$ can be arbitrarily large. \ In this version, we've
reworked the lower bound to incorporate the dependence on $D$ explicitly.
\ One side effect is that, if we let $M$ be arbitrarily large, then our lower
bound now subsumes the $\Omega\left(  D^{2}\right)  $\ lower bound on the
sample complexity of \textit{ordinary} quantum state tomography, originally
proved by O'Donnell and Wright \cite{owright}\ and Haah et al.\ \cite{hhjwy}.}
\ We'll prove this lower bound in Section \ref{LBAPPENDIX}, using an
information theory argument. \ We'll also observe that a lower bound of
$\Omega\left(  \frac{\min\left\{  D,\log M\right\}  }{\varepsilon^{2}}\right)
$ holds even in the special case where the state and measurements are entirely
classical---in which case the lower bound is actually \textit{tight}. \ In the
general (quantum) case, we don't know whether shadow tomography might be
possible using $\left(  \log M\right)  ^{O\left(  1\right)  }$ copies,
independent of the Hilbert space dimension $D$.

But stepping back, why isn't even Theorem \ref{main} immediately ruled out by,
for example, Holevo's Theorem \cite{holevo}---which says (roughly) that by
measuring a $D$-dimensional state, we can learn at most $O\left(  \log
D\right)  $ independent classical bits? \ One way to answer this question is
to observe that there's no claim that the $M$ numbers $\operatorname{Tr}%
\left(  E_{i}\rho\right)  $\ can all be varied independently of each other by
varying $\rho$: indeed, it follows from known results \cite{antv,aar:learn}%
\ that they can't be, unless $D=\exp\left(  \Omega\left(  M\right)  \right)  $.

Another answer is as follows. \ It's true that there exist so-called
\textit{tomographically complete} sets of two-outcome measurements, of size
$M=O\left(  D^{2}\right)  $. \ These are sets $E_{1},\ldots,E_{M}$\ such that
knowing $\operatorname{Tr}\left(  E_{i}\rho\right)  $\ exactly, for every
$i\in\left[  M\right]  $, suffices to determine $\rho$\ itself. \ So if we ran
our shadow tomography procedure on a tomographically complete set, with small
enough $\varepsilon$, then we could reconstruct $\rho$, something that we know
requires $k=\Omega\left(  D^{2}\right)  $\ copies of $\rho$. \ However, this
would require knowing the $\operatorname{Tr}\left(  E_{i}\rho\right)  $'s to
within additive error $\varepsilon\ll1/D$, which remains perfectly compatible
with a shadow tomography procedure that uses $\operatorname*{poly}\left(  \log
M,\log D,\varepsilon^{-1}\right)  $ copies.

One last clarifying remark is in order. \ After satisfying themselves that
it's not impossible, some readers might wonder whether Theorem \ref{main}
follows trivially from the so-called \textquotedblleft Gentle Measurement
Lemma\textquotedblright\ \cite{winter:gentle,aar:adv}, which is closely
related to the concept of \textit{weak measurement} in physics. \ We'll
explain gentle measurement in more detail in Section \ref{PRELIM}, but loosely
speaking, the idea is that \textit{if} the outcome of a measurement $E$\ on a
state $\rho$\ could be predicted almost with certainty, given knowledge of
$\rho$, \textit{then} $E$ can be implemented in a way that damages $\rho
$\ very little, leaving the state available for future measurements. \ Gentle
measurement will play an important role in the proof of Theorem \ref{main}, as
it does in many quantum information results.

However, all that we can \textit{easily} deduce from gentle measurement is a
\textquotedblleft promise-gap\textquotedblright\ version of Theorem
\ref{main}. \ In particular: suppose we're given real numbers $c_{1}%
,\ldots,c_{M}\in\left[  0,1\right]  $, and are promised that for each
$i\in\left[  M\right]  $, either $\operatorname{Tr}\left(  E_{i}\rho\right)
\geq c_{i}$\ or $\operatorname{Tr}\left(  E_{i}\rho\right)  \leq
c_{i}-\varepsilon$. \ In that case, we'll state and prove, as Proposition
\ref{gap}, that it's possible to decide which of these holds, for every
$i\in\left[  M\right]  $, with high probability using only $k=O\left(
\frac{\log M}{\varepsilon^{2}}\right)  $\ copies of $\rho$. \ This is because,
\textit{given the promise gap}, we can design an \textquotedblleft
amplified\textquotedblright\ version of $E_{i}$\ that decides which side of
the gap we're on while damaging $\rho^{\otimes k}$\ only very little.

But what if there's no promise, as there typically isn't in real-world
tomography problems? \ In that case, the above approach fails utterly: indeed,
every two-outcome measurement $E$\ that we could possibly apply\ seems
dangerous, because if $\rho$\ happens to be \textquotedblleft just on the
knife-edge\textquotedblright\ between acceptance and rejection---a possibility
that we can never rule out---then applying $E$\ to copies of $\rho$\ will
severely damage those copies. \ And while we can afford to lose a \textit{few}
copies of $\rho$, we have only $\operatorname*{poly}\left(  \log M,\log
D\right)  $\ copies in total, which is typically far fewer than the
$M$\ measurement outcomes that we need to learn.\footnote{As an alternative,
one might hope to prove Theorem \ref{main}\ by simply performing a series of
\textquotedblleft weak measurements\textquotedblright\ on the state
$\rho^{\otimes k}$, which would estimate the real-valued observables
$\operatorname{Tr}\left(  E_{i}\rho\right)  $, but with Gaussian noise of
variance $\gg1/k$\ deliberately added to the measurement outcomes, in order to
prevent $\rho^{\otimes k}$\ from being damaged too much by the measurements.
\ However, a calculation reveals that every such measurement could damage the
state by $1/k^{O\left(  1\right)  }$\ in variation distance. \ Thus, while
this strategy would let us safely estimate $\operatorname*{poly}\left(  \log
M,\log D\right)  $\ observables $\operatorname{Tr}\left(  E_{i}\rho\right)
$\ in succession, it doesn't appear to let us estimate all $M$ of them.}
\ This is the central problem that we solve.

\subsection{Techniques\label{TECHNIQUES}}

At a high level, our shadow tomography procedure involves combining two ideas.

The first idea is \textit{postselected learning of quantum states}. \ This
tool was introduced by Aaronson \cite{aar:adv} in 2004 to prove the complexity
class containment $\mathsf{BQP/qpoly}\subseteq\mathsf{PostBQP/poly}$, where
$\mathsf{BQP/qpoly}$\ means $\mathsf{BQP}$\ augmented with polynomial-size
quantum advice, and $\mathsf{PostBQP}$\ means $\mathsf{BQP}$\ augmented with
postselected measurements,\ a class that equals $\mathsf{PP}$\ by another
result of Aaronson \cite{aar:pp}. \ Postselected learning is related to
\textit{boosting} in computational learning theory, as well as to the
multiplicative weights update method.

Restated in the language of this paper, the canonical example of postselected
learning is as follows. \ Suppose Alice knows the complete classical
description of a $D$-dimensional quantum mixed state $\rho$, and suppose she
wants to describe $\rho$\ to Bob over a classical channel---well enough that
Bob can approximate the value of $\operatorname{Tr}\left(  E_{i}\rho\right)
$, for each of $M$ two-outcome measurements $E_{1},\ldots,E_{M}$\ known to
both players. \ To do this, Alice could always send over the full classical
description of $\rho$, requiring $\Theta\left(  D^{2}\right)  $\ bits. \ Or
she could send the values of the\ $\operatorname{Tr}\left(  E_{i}\rho\right)
$'s, requiring $\Theta\left(  M\right)  $ bits.

But there's also something much more efficient that Alice can do, requiring
only $\Theta\left(  \log D\cdot\log M\right)  $ bits. \ Namely, she can assume
that, being totally ignorant at first, Bob's \textquotedblleft initial
guess\textquotedblright\ about $\rho$\ is simply that it's the maximally mixed
state, $\rho_{0}:=\frac{I}{D}$. \ She can then repeatedly help Bob to refine
his current guess $\rho_{t}$ to a better guess $\rho_{t+1}$, by telling Bob
the index $i$\ of a measurement on which his current guess badly fails---that
is, on which $\left\vert \operatorname{Tr}\left(  E_{i}\rho_{t}\right)
-\operatorname{Tr}\left(  E_{i}\rho\right)  \right\vert $\ is large---as well
as the approximate value of $\operatorname{Tr}\left(  E_{i}\rho\right)  $.
\ To use this information, Bob can let $\rho_{t+1}$\ be the state obtained by
starting from $\rho_{t}$\ (or technically, an amplified version of $\rho_{t}%
$), measuring the observable $E_{i}$, and then \textit{postselecting} (that
is, conditioning) on getting measurement outcomes that are consistent with
$\rho$. \ Of course this postselection might have only a small chance of
success, were Bob doing it with the actual state $\rho_{t}$, but he can
instead \textit{simulate} postselection using a classical description of
$\rho_{t}$.

The key question, with this approach, is how many iterations $T$\ are needed
until Bob converges to a hypothesis state $\rho_{T}$\ such that
$\operatorname{Tr}\left(  E_{i}\rho_{T}\right)  \approx\operatorname{Tr}%
\left(  E_{i}\rho\right)  $ for every $i$. \ And the key result is that only
$\Theta\left(  \log D\right)  $\ iterations are needed. \ Intuitively, this is
because the ground truth, $\rho$, has \textquotedblleft
weight\textquotedblright\ at least $\frac{1}{D}$\ within the maximally mixed
state $\frac{I}{D}$. \ Repeatedly choosing measurements where the current
hypothesis still does poorly, and then postselecting on doing well on those
measurements, causes all the components of $\frac{I}{D}$\ \textit{other than}
$\rho$ to decay at an exponential rate, until a measurement can no longer be
found where the current hypothesis does poorly. \ That might happen well
before we reach $\rho$\ itself, but if not, then $\rho$\ itself will be
reached after $\Theta\left(  \log D\right)  $\ iterations.

Postselected learning has since found further uses in quantum computing theory
\cite{aar:qmaqpoly,achn}. \ But there seems to be a fundamental difficulty in
applying it to shadow tomography. \ Namely, in shadow tomography\textit{
there's no \textquotedblleft Alice\textquotedblright}: that is, no agent who
knows a classical description of the state $\rho$, and who can thus helpfully
point to measurements $E_{i}$\ that are useful for learning $\rho$'s behavior.
\ So any shadow tomography procedure will need to find informative
measurements by itself, and do so using only polylogarithmically many copies
of $\rho$.\bigskip

The second idea, the \textit{gentle search procedure}, does exactly that. \ In
2006, as a central ingredient in the proof of the complexity class containment
$\mathsf{QMA/qpoly}\subseteq\mathsf{PSPACE/poly}$, Aaronson
\cite{aar:qmaqpoly}\ claimed a result that he called \textquotedblleft Quantum
OR Bound.\textquotedblright\ \ This result can be stated as follows:\ given an
unknown state $\rho$\ and known two-outcome measurements $E_{1},\ldots,E_{M}$,
there is a procedure, using $k=O\left(  \frac{\log M}{\varepsilon^{2}}\right)
$\ copies of $\rho$, to decide whether

\begin{enumerate}
\item[(i)] some $E_{i}$\ accepts $\rho$\ with probability at least $c$ or

\item[(ii)] no $E_{i}$\ accepts $\rho$\ with probability greater than
$c-\varepsilon$,
\end{enumerate}

\noindent with high probability and assuming one of the cases holds. \ Note
that the number of copies is not only logarithmic in $M$, but independent of
the dimension of $\rho$.

Aaronson's proof of the Quantum OR Bound was based on simply applying
amplified versions of the $E_{i}$'s to $\rho^{\otimes k}$\ in a random order,
and checking whether any of the measurements accepted. \ Unfortunately,
Aaronson's proof had an error, which was discovered in 2016 by Harrow, Lin,
and Montanaro \cite{hlm}. \ Happily, Harrow et al.\ also fixed the error,
thereby recovering all the consequences that Aaronson had claimed, as well as
new consequences. \ To do so, Harrow et al.\ designed two new measurement
procedures, both of which solve the problem: one based on the
\textquotedblleft in-place amplification\textquotedblright\ of Marriott and
Watrous \cite{mw}, and another that applies amplified $E_{i}$'s conditional on
a control qubit being $\left\vert 1\right\rangle $, and that checks not only
whether any of the measurements accept but also whether the control qubit has
decohered. \ It remains open whether Aaronson's original procedure is also sound.

For shadow tomography, however, there's a further problem. \ Namely, at each
iteration of the postselected learning procedure, we need not only to decide
whether there \textit{exists} an $i$\ such that $\left\vert \operatorname{Tr}%
\left(  E_{i}\rho_{t}\right)  -\operatorname{Tr}\left(  E_{i}\rho\right)
\right\vert $\ is large, but also to \textit{find} such an $i$ if it exists.
\ Fortunately, we can handle this using the \textquotedblleft oldest trick in
the book\textquotedblright\ for reducing search problems to decision problems:
namely, binary search over the list $E_{1},\ldots,E_{M}$. \ Doing this
correctly requires carefully managing the error budget---as we proceed through
binary search, the gap between $\operatorname{Tr}\left(  E_{i}\rho_{t}\right)
$ and $\operatorname{Tr}\left(  E_{i}\rho\right)  $\ that we're confident
we've found degrades from $\varepsilon$\ to $\varepsilon-\alpha$\ to
$\varepsilon-2\alpha$, etc.---and that's what produces the factor of $\log
^{4}M$\ in the final bound.

\subsection{Comparison with Related Work\label{RELATED}}

While we've already discussed a good deal of related work, here we'll compare
Theorem \ref{main}\ directly against some previous results, and explain why
those results fall short of what we need. \ We'll then discuss the recent work
of Brand\~{a}o et al.\ \cite{bkllsw}, which builds on this paper to address
the computational cost of shadow tomography.

One important inspiration for what we're trying to do, and something we
\textit{haven't} yet discussed, is the \textquotedblleft Quantum Occam's Razor
Theorem,\textquotedblright\ which Aaronson \cite{aar:learn} proved in 2006.
\ This result essentially says that quantum states are \textquotedblleft
learnable\textquotedblright\ in the PAC (Probably Approximately Correct) sense
\cite{valiant:pac}, with respect to any probability distribution over
two-outcome measurements, using an amount of sample data that increases only
\textit{linearly} with the number of qubits---rather than exponentially, as
with traditional quantum state tomography. \ More formally:

\begin{theorem}
[Quantum Occam's Razor \cite{aar:learn}]\label{learnthm}Let $\rho$\ be a
$D$-dimensional mixed state, and let $\mu$\ be any probability distribution or
measure over two-outcome measurements. \ Then given samples $E_{1}%
,\ldots,E_{M}$ drawn independently from $\mu$, with probability at least
$1-\delta$, the samples have the following generalization property: any
hypothesis state $\sigma$\ such that $\left\vert \operatorname{Tr}\left(
E_{i}\sigma\right)  -\operatorname{Tr}\left(  E_{i}\rho\right)  \right\vert
\leq\frac{\gamma\varepsilon}{7}$ for all $i\in\left[  M\right]  $, will also
satisfy%
\[
\Pr_{E\sim\mu}\left[  \left\vert \operatorname{Tr}\left(  E\sigma\right)
-\operatorname{Tr}\left(  E\rho\right)  \right\vert \leq\varepsilon\right]
\geq1-\gamma,
\]
provided we took%
\[
M\geq\frac{C}{\gamma^{2}\varepsilon^{2}}\left(  \frac{\log D}{\gamma
^{2}\varepsilon^{2}}\log^{2}\frac{1}{\gamma\varepsilon}+\log\frac{1}{\delta
}\right)
\]
for some large enough constant $C$.
\end{theorem}

We could try applying Theorem \ref{learnthm} to the shadow tomography problem.
\ If we do, however, we get only that $\widetilde{O}\left(  \frac{\log
D}{\gamma^{4}\varepsilon^{4}}\right)  $\ copies of $\rho$\ are enough to let
us estimate $\operatorname{Tr}\left(  E_{i}\rho\right)  $\ to within error
$\pm\varepsilon$, on at least a $1-\gamma$\ \textit{fraction} of the
measurements $E_{1},\ldots,E_{M}$---rather than on \textit{all} the measurements.

If we want a result that works for all $E_{i}$'s, then we can instead switch
attention to Aaronson's postselected learning theorem \cite{aar:adv}, the one
used to prove the containment $\mathsf{BQP/qpoly}\subseteq
\mathsf{PostBQP/poly}$. \ For completeness, let us restate that theorem\ in
the language of this paper.

\begin{theorem}
[implicit in \cite{aar:adv}; see also \cite{achn}]\label{postthm}Let $\rho
$\ be an unknown $D$-dimensional mixed state, and\ let $E_{1},\ldots,E_{M}%
$\ be known two-outcome measurements. \ Then there exists a classical string,
of length $\widetilde{O}\left(  \frac{\log D\cdot\log M}{\varepsilon^{3}%
}\right)  $, from which $\operatorname{Tr}\left(  E_{i}\rho\right)  $\ can be
recovered to within additive error $\pm\varepsilon$\ for every $i\in\left[
M\right]  $.
\end{theorem}

As we mentioned in Section \ref{TECHNIQUES}, Theorem \ref{postthm}\ falls
short of shadow tomography simply because it's \textquotedblleft
nondeterministic\textquotedblright: it says that a short classical string
\textit{exists} from which one could recover the approximate values of every
$\operatorname{Tr}\left(  E_{i}\rho\right)  $, but says nothing about how to
find such a string by measuring few copies of $\rho$.

There's a different way to think about Theorem \ref{postthm}. \ Along the way
to proving the containment $\mathsf{BQP/qpoly}\subseteq\mathsf{QMA/poly}$,
Aaronson and Drucker \cite{adrucker}\ observed the following, by combining a
result from classical learning theory\ with a result from \cite{aar:learn}%
\ about the \textquotedblleft fat-shattering dimension\textquotedblright\ of
quantum states as a hypothesis class.

\begin{theorem}
[Aaronson and Drucker \cite{adrucker}]\label{epscover}Let $E_{1},\ldots,E_{M}%
$\ be two-outcome measurements on $D$-dimensional Hilbert space. \ Then there
exists a set $S$ of real functions $f:\left[  M\right]  \rightarrow\left[
0,1\right]  $, of cardinality $\left(  \frac{M}{\varepsilon}\right)
^{O\left(  \varepsilon^{-2}\log D\right)  }$, such that for every
$D$-dimensional mixed state $\rho$, there exists an $f\in S$\ such that
$\left\vert f\left(  i\right)  -\operatorname{Tr}\left(  E_{i}\rho\right)
\right\vert \leq\varepsilon$\ for all $i\in\left[  M\right]  $.
\end{theorem}

Up to a small difference in the parameters, Theorem \ref{epscover}\ is
equivalent to Theorem \ref{postthm}: either can easily be deduced from the
other. \ The main difference is just that Theorem \ref{postthm} came with an
explicit procedure, based on postselection, for recovering the
$\operatorname{Tr}\left(  E_{i}\rho\right)  $'s from the classical string,
whereas Theorem \ref{epscover} was much less explicit.

It might seem that Theorem \ref{epscover} would give rise to a shadow
tomography procedure, since we'd just need to implement a measurement, say on
$O\left(  \frac{\log\left\vert S\right\vert }{\varepsilon^{2}}\right)  $
copies of $\rho$, that \textquotedblleft pulled apart\textquotedblright\ the
different elements of the set $S$\ (which is called\ an $\varepsilon
$\textit{-cover}). \ Unfortunately, we haven't been able to turn this
intuition into an algorithm. \ For while one \textit{can} project a quantum
state onto any set of vectors that's sufficiently close to orthogonal---as,
for example, in the algorithm of Ettinger, H\o yer, and Knill \cite{ehk}\ for
the hidden subgroup problem---in shadow tomography, there's no guarantee that
the state $\rho^{\otimes k}$\ being measured \textit{is} close to one of
various nearly-orthogonal measurement outcomes, and therefore that it won't be
irreparably damaged at an early stage in the measurement process.\bigskip

Recently, building on the work reported here, Brand\~{a}o et
al.\ \cite{bkllsw} have undertaken an initial investigation of the
\textit{computational} complexity of shadow tomography. \ While we made no
attempt to optimize the computational cost of our procedure, a loose estimate
is that ours requires performing $\widetilde{O}\left(  \frac{M\log
D}{\varepsilon^{4}}\right)  $ measurements on copies of $\rho$. \ Furthermore,
each measurement itself could, in the worst case, require $\Theta\left(
D^{2}\right)  $ gates to implement. \ Our procedure also involves storing and
updating a classical description of an amplified hypothesis state, which takes
$D^{O\left(  \varepsilon^{-2}\log\log D\right)  }$\ time and space.

By combining our ideas with recent quantum algorithms for semidefinite
programming, Brand\~{a}o et al.\ \cite{bkllsw} have shown how to perform
shadow tomography using not only $\operatorname*{poly}\left(  \log M,\log
D\right)  $\ copies of $\rho$, but also $\widetilde{O}\left(  \sqrt
{M}L\right)  +D^{O\left(  1\right)  }$ quantum gates, where $L=O\left(
D^{2}\right)  $\ is the maximum length of a circuit to apply a single
measurement $E_{i}$. \ This of course improves over our $\widetilde{O}\left(
ML\right)  +D^{O\left(  \log\log D\right)  }$.

If we make some additional assumptions about the measurement matrices $E_{i}%
$---namely, that they\ have rank at most $\operatorname*{polylog}D$; and that
for every $i\in\left[  M\right]  $, one can coherently prepare the mixed state
$\frac{E_{i}}{\operatorname{Tr}\left(  E_{i}\right)  }$, and also compute
$\operatorname{Tr}\left(  E_{i}\right)  $, in time at most
$\operatorname*{polylog}D$---then Brand\~{a}o et al.\ \cite{bkllsw}\ further
improve the running time of their algorithm, to $\widetilde{O}\left(  \sqrt
{M}\operatorname*{polylog}D\right)  $.

Roughly, Brand\~{a}o et al.\ \cite{bkllsw} keep much of the structure of our
algorithm, except they replace\ our linear search for informative measurements
$E_{i}$\ by Grover-style approximate counting---hence the improvement from
$\widetilde{O}\left(  M\right)  $\ to $\widetilde{O}\left(  \sqrt{M}\right)
$. \ They also replace our postselected learning by the preparation of a Gibbs
state, using Jaynes' principle from statistical mechanics. \ By exploiting
recent progress on quantum algorithms for SDPs, Brand\~{a}o et al.\ are able
to perform the needed manipulations on $D$-dimensional hypothesis states
without ever writing the states explicitly in a classical memory as $D\times
D$\ matrices, like we do.

In Section \ref{OPEN}, we'll discuss the prospects for improving the gate
complexity of shadow tomography further, and some possible
complexity-theoretic barriers to doing so.\bigskip

There are many other results in the literature that can be seen, in one way or
another, as trying to get around the destructive nature of measurement, or the
exponential number of copies needed for state tomography. \ We won't even
attempt a survey here, but briefly, such results often put some additional
restriction on the state $\rho$\ to be learned: for example, that it's low
rank \cite{glfbe}, or that it has a succinct classical description of some
kind (e.g., that it's a stabilizer state \cite{montanaro:stabilizer}), or that
we have an oracle to recognize the state \cite{farhi:restore}. \ Of course,
shadow tomography requires none of these assumptions.

\section{Motivation\label{MOTIV}}

Perhaps the most striking way to state Theorem \ref{main} is as follows.

\begin{corollary}
\label{circuitcor}Let $\left\vert \psi\right\rangle $\ be an unknown $n$-qubit
state, and let $p$ be any fixed polynomial. \ Then it's possible to estimate
$\Pr\left[  C\text{ accepts }\left\vert \psi\right\rangle \right]  $\ to
within additive error $\pm\varepsilon$, for \textbf{every} quantum circuit $C$
with at most $p\left(  n\right)  $ gates simultaneously, and with $1-o\left(
1\right)  $\ success probability, by a measurement on $\left(  n/\varepsilon
\right)  ^{O\left(  1\right)  }$\ copies of $\left\vert \psi\right\rangle $.
\end{corollary}

Here, we're simply combining Theorem \ref{main}\ with the observation that
there are at most $M=\left(  n+p\left(  n\right)  \right)  ^{O\left(  p\left(
n\right)  \right)  }$ different quantum circuits of size at most $p\left(
n\right)  $, assuming a fixed finite gate set without loss of generality.

We've already given some philosophical motivation for this: at bottom we're
trying to understand, \textit{to what extent does the destructive nature of
quantum measurement force us into an epistemically unsatisfying situation,
where we need }$\exp\left(  n\right)  $\textit{\ copies of an }$\mathit{n}%
$\textit{-qubit state }$\left\vert \psi\right\rangle $\textit{\ just to learn
}$\left\vert \psi\right\rangle $\textit{'s basic properties? \ }Corollary
\ref{circuitcor} tells us that, as long as the \textquotedblleft basic
properties\textquotedblright\ are limited to $\left\vert \psi\right\rangle $'s
accept/reject behaviors on quantum circuits of a fixed polynomial size (and to
whatever can be deduced from those behaviors), we're \textit{not} in the
epistemically unsatisfying situation that might have been feared.

Besides this conceptual point, we hope that Theorem \ref{main} will
find\ experimental applications. \ In the quest for such applications, it
would of course help to tighten the parameters of Theorem \ref{main} (e.g.,
the exponents in $\frac{\log^{4}M}{\varepsilon^{5}}$); and to find shadow
tomography procedures that are less expensive both in computational complexity
and in the required measurement apparatus. \ We'll say more about these issues
in Section \ref{OPEN}.

In the rest of this section, we'll point out implications of Theorem
\ref{main}\ for several areas of quantum computing theory: quantum money,
quantum copy-protected software, and quantum advice and one-way communication.
\ The first of these actually provided the original impetus for this work: as
we'll explain, Theorem \ref{main} immediately yields a proof of a basic result
called the \textquotedblleft tradeoff theorem\textquotedblright\ for
private-key quantum money schemes \cite[Section 8.3]{aarbados}. \ But even
where the implications amount to little more than translations of the theorem
to other contexts, they illustrate the wide reach of shadow tomography as a
concept.\bigskip

\textbf{Quantum money.} \ The idea of quantum money---i.e., quantum states
that can be traded and verified, but are physically impossible to clone---is
one of the oldest ideas in quantum information, having been proposed by
Wiesner \cite{wiesner} around 1970. \ A crucial distinction here is between
so-called \textit{public-key} and \textit{private-key} quantum money schemes.
\ See Aaronson and Christiano \cite{achristiano}\ for formal definitions of
these concepts, but briefly: in a public-key money scheme, anyone can
efficiently verify a bill $\left\vert \$\right\rangle $\ as genuine, whereas
in a private-key scheme, verifying a bill requires taking it back to the bank.
\ It's easy to see that, if public-key quantum money is possible at all, then
it requires computational assumptions (e.g., that any would-be counterfeiter
is limited to polynomial time). \ While Aaronson and Christiano
\cite{achristiano} constructed an oracle relative to which public-key quantum
money is possible, it's still unclear whether it's possible in the
unrelativized world.

By contrast, in Wiesner's original paper on the subject \cite{wiesner}, he
proposed a private-key quantum money scheme that was \textit{unconditionally
secure} (though a security proof would only be given in 2012, by Molina,
Vidick, and Watrous \cite{molina}). \ The central defect of Wiesner's scheme
was that it required the bank to maintain a gigantic database, storing a
different list of secret measurement bases for every bill in circulation. \ In
1982, Bennett et al.\ \cite{bbbw}\ fixed this defect of Wiesner's scheme, but
only by using a pseudorandom function to generate the measurement bases---so
that the scheme again required a computational assumption.

In 2009, Aaronson \cite{aar:qcopy} raised the question of whether there's an
inherent tradeoff here: that is, does every private-key quantum money scheme
require \textit{either} a huge database, or else a computational
assumption?\footnote{Actually, he claimed to have an unwritten proof of this,
but working out the details took longer than expected, and indeed ultimately
relied on the 2016 work of Harrow, Lin, and Montanaro \cite{hlm}.} \ He then
answered this question in the affirmative (paper still in preparation, but see
\cite[Section 8.3]{aarbados}). \ It was while proving this tradeoff theorem
that the author was led to formulate the shadow tomography problem.

To see the connection, let's observe an easy corollary of Theorem \ref{main}.

\begin{corollary}
[of Theorem \ref{main}]\label{moneycor}Consider any private-key quantum money
scheme with a single secret key $k\in\left\{  0,1\right\}  ^{m}$ held by the
bank; $d$-qubit bills\ $\left\vert \$\right\rangle $; and a verification
procedure $V\left(  k,\left\vert \$\right\rangle \right)  $\ that the bank
applies. \ Then given $\widetilde{O}\left(  dm^{4}\right)  $\ legitimate bills
$\left\vert \$\right\rangle $, as well as $\exp\left(  d,m\right)  $
computation time, a counterfeiter can estimate $\Pr\left[  V\left(
k,\left\vert \$\right\rangle \right)  \text{ accepts}\right]  $\ to within
additive error $o\left(  1\right)  $, for every $k\in\left\{  0,1\right\}
^{m}$, with success probability $1-o\left(  1\right)  $.
\end{corollary}

\begin{proof}
We set $M:=2^{m}$, and let our list of $M$\ two-outcome measurements
correspond to $V\left(  k,\cdot\right)  $ for every $k\in\left\{  0,1\right\}
^{m}$. \ We set $\rho:=\left\vert \$\right\rangle \left\langle \$\right\vert
$; this is a $D$-dimensional state where $D:=2^{d}$. \ Then Theorem
\ref{main}\ lets us estimate $\Pr\left[  V\left(  k,\left\vert \$\right\rangle
\right)  \text{ accepts}\right]  $\ for every $k\in\left\{  0,1\right\}  ^{m}%
$\ as claimed, using%
\[
\widetilde{O}\left(  \log D\cdot\log^{4}M\right)  =\widetilde{O}\left(
dm^{4}\right)
\]
copies of $\left\vert \$\right\rangle $.
\end{proof}

We now observe that the tradeoff theorem follows immediately from Corollary
\ref{moneycor}:

\begin{theorem}
[Tradeoff Theorem for Quantum Money]\label{tradeoff}Given any private-key
quantum money scheme, with $d$-qubit bills\ and an $m$-bit secret key held by
the bank, a counterfeiter can produce additional bills, which pass
verification with $1-o\left(  1\right)  $\ probability, given $\widetilde{O}%
\left(  dm^{4}\right)  $\ legitimate bills and $\exp\left(  d,m\right)  $
computation time. \ No queries to the bank are needed to produce these bills.
\end{theorem}

For given exponential time, the counterfeiter just needs to do a brute-force
search (for example, using semidefinite programming) for a state $\rho$\ such
that%
\[
\left\vert \Pr\left[  V\left(  k,\rho\right)  \text{ accepts}\right]
-\Pr\left[  V\left(  k,\left\vert \$\right\rangle \right)  \text{
accepts}\right]  \right\vert =o\left(  1\right)
\]
for every key $k\in\left\{  0,1\right\}  ^{m}$. \ Such a $\rho$\ surely
exists, since $\left\vert \$\right\rangle $\ itself is one, and given
exponential time, the counterfeiter can then prepare $\rho$\ as often as it
likes. \ And by assumption, this $\rho$ must be a state that the bank accepts
with high probability given the \textquotedblleft true\textquotedblright\ key
$k^{\ast}$---\textit{even though the counterfeiter never actually learns
}$k^{\ast}$\textit{\ itself}.

In \cite[Section 8.3]{aarbados}, the author took a somewhat different route to
proving the tradeoff theorem, simply because he didn't yet possess the shadow
tomography theorem. \ Specifically, he used what in this paper we'll call the
\textquotedblleft gentle search procedure,\textquotedblright\ and will prove
as Lemma \ref{searchlem} along the way to proving Theorem \ref{main}. \ He
then combined Lemma \ref{searchlem}\ with an iterative procedure, which
repeatedly cut down the space of \textquotedblleft possible
keys\textquotedblright\ $k$ by a constant factor, until averaging over the
remaining keys led to a state that the bank accepted with high probability.
\ However, this approach had the drawback that preparing the counterfeit bills
required $O\left(  n\right)  $\ queries to the bank. \ Shadow
tomography\ removes that drawback.\bigskip

\textbf{Quantum copy-protected software.} \ In 2009, Aaronson \cite{aar:qcopy}
introduced the notion of quantum copy-protected software: roughly speaking, an
$n^{O\left(  1\right)  }$-qubit quantum state $\rho_{f}$\ that's given to a
user, and that lets the user efficiently evaluate a Boolean function
$f:\left\{  0,1\right\}  ^{n}\rightarrow\left\{  0,1\right\}  $, on any input
$x\in\left\{  0,1\right\}  ^{n}$\ of the user's choice, but that can't be used
to prepare more states with which $f$\ can be efficiently evaluated. \ The
analogous classical problem is clearly impossible. \ But the destructive
nature of quantum measurements (or equivalently, the unclonability of quantum
states) raises the prospect that, at least with suitable cryptographic
assumptions, it could be possible quantumly. \ And indeed, Aaronson
\cite{aar:qcopy} sketched a construction of a quantum oracle $U$ relative to
which quantum copy-protection is \textquotedblleft
generically\textquotedblright\ possible, meaning that one really \textit{can}
have a state $\left\vert \psi_{f}\right\rangle $\ that acts like an unclonable
black box for any Boolean function $f$ of one's choice. \ It remains an
outstanding problem to construct \textit{explicit} schemes for quantum
copy-protection, which are secure under plausible cryptographic assumptions.

But now suppose that we're interested in quantum programs that simply accept
various inputs $x\in\left\{  0,1\right\}  ^{n}$ with specified probabilities
$p\left(  x\right)  \in\left[  0,1\right]  $: for example, programs to
evaluate partial Boolean functions, or to simulate quantum processes. \ In
that case, we might hope for a copy-protection scheme that was
\textit{unconditionally} secure, even against software pirates with unlimited
computation time. \ Furthermore, such a scheme would have the
property---possibly desirable to the software vendor!---that the
programs\ would periodically get \textquotedblleft used up\textquotedblright%
\ even by legitimate use, and need to be replenished. \ For even if we had
$n^{O\left(  1\right)  }$\ copies of the program, and used the Gentle
Measurement Lemma to estimate the probabilities $p\left(  x\right)  $, we
still couldn't always avoid measurements on the \textquotedblleft knife
edge\textquotedblright\ between one output behavior and another, which would
destroy the copies.

Once again, though, Theorem \ref{main} has the consequence that this gambit
fails, so that if quantum copy-protection is possible at all, then it indeed
requires computational assumptions.

\begin{corollary}
[of Theorem \ref{main}]\label{copycor}Let $\rho$\ be any $n^{O\left(
1\right)  }$-qubit quantum program, which accepts each input $x\in\left\{
0,1\right\}  ^{n}$ with probability $p\left(  x\right)  $. \ Then given
$n^{O\left(  1\right)  }$\ copies of $\rho$\ and $2^{n^{O\left(  1\right)  }}%
$\ computation time, with $1-o\left(  1\right)  $\ success probability we can
\textquotedblleft pirate\textquotedblright\ $\rho$: that is, produce multiple
quantum programs, all of which accept input $x\in\left\{  0,1\right\}  ^{n}$
with probability $p\left(  x\right)  \pm o\left(  1\right)  $, and which have
the same running time as $\rho$\ itself.
\end{corollary}

Here we're using the fact that, once we know the approximate acceptance
probabilities of $\rho$\ on every input $x\in\left\{  0,1\right\}  ^{n}$, in
$2^{n^{O\left(  1\right)  }}$\ time we can simply use semidefinite programming
to brute-force search for an $n^{O\left(  1\right)  }$-qubit state $\sigma
$\ that approximates $\rho$'s acceptance probabilities on every $x$. \ Indeed,
if we further assume that $\rho$\ was prepared by a polynomial-size quantum
circuit, then in $2^{n^{O\left(  1\right)  }}$\ time\ we can brute-force
search for such a circuit as well.\bigskip

\textbf{Quantum advice and one-way communication.} \ In 2003, Nishimura and
Yamakami \cite{ny} defined the complexity class $\mathsf{BQP/qpoly}$, which
consists (informally) of all languages that are decidable in bounded-error
quantum polynomial time, given a polynomial-size \textquotedblleft quantum
advice state\textquotedblright\ $\left\vert \psi_{n}\right\rangle $\ that
depends only on the input length $n$ but could otherwise be arbitrary. \ This
is a natural quantum generalization of the classical notion of Karp-Lipton
advice, and of the class $\mathsf{P/poly}$. \ Many results have since been
proven about $\mathsf{BQP/qpoly}$ and related classes
\cite{aar:adv,aar:qmaqpoly,adrucker,ak}; and as we discussed in Section
\ref{TECHNIQUES}, some of the techniques used to prove those results will also
play major roles in this work.

But one basic question remained: given a $\mathsf{BQP/qpoly}$\ algorithm,
suppose we're given $n^{O\left(  1\right)  }$\ copies of the quantum advice
state $\left\vert \psi_{n}\right\rangle $. \ Can we safely reuse those copies,
again and again, for as many inputs $x\in\left\{  0,1\right\}  ^{n}$\ as we
like? \ For \textit{deciding a language} $L$, it's not hard to show that the
answer is yes, because of the Gentle Measurement Lemma\ (Lemma \ref{gentle} in
Section \ref{PRELIM}). \ But if we consider \textit{promise problems} (i.e.,
problems of deciding which of two disjoint sets the input $x$\ belongs to,
promised that it belongs to one of them), then a new difficulty arises.
\ Namely, what if we use our quantum advice on an input that violates the
promise---a possibility that we can't generally avoid if we don't know the
promise? \ Every such use runs the risk of destroying an advice state.

An immediate corollary of Theorem \ref{main} is that we can handle this issue,
albeit with a blowup in computation time.

\begin{corollary}
[of Theorem \ref{main}]\label{promisecor}Let $\Pi=\left(  \Pi
_{\operatorname*{YES}},\Pi_{\operatorname*{NO}}\right)  $ be a promise problem
in $\mathsf{P{}romiseBQP/qpoly}$. \ Let $A$ be a quantum algorithm for $\Pi
$\ that uses advice states $\left\{  \left\vert \psi_{n}\right\rangle
\right\}  _{n}$. \ Then there exists a quantum algorithm, running in
$2^{n^{O\left(  1\right)  }}$\ time, that uses $\left\vert \psi_{n}%
\right\rangle ^{\otimes n^{O\left(  1\right)  }}$ as advice, and that
approximates $\Pr\left[  A\left(  x,\left\vert \psi_{n}\right\rangle \right)
\text{ accepts}\right]  $\ to within $\pm o\left(  1\right)  $, for all
$2^{n}$ inputs $x\in\left\{  0,1\right\}  ^{n}$, with success probability
$1-o\left(  1\right)  $. \ So in particular, this algorithm \textquotedblleft
generates the complete truth table of $\Pi$\ on inputs of size $n$%
,\textquotedblright\ and does so even without being told which inputs satisfy
the promise $x\in\Pi_{\operatorname*{YES}}\cup\Pi_{\operatorname*{NO}}$.
\end{corollary}

We can also state Corollary \ref{promisecor} in terms of \textit{quantum
one-way communication protocols}. \ In that case, the corollary says the
following. \ Suppose Alice holds an input $x\in\left\{  0,1\right\}  ^{n}%
$\ and Bob holds an input $y\in\left\{  0,1\right\}  ^{m}$, and they want to
compute a partial Boolean function $f:S\rightarrow\left\{  0,1\right\}  $, for
some $S\subset\left\{  0,1\right\}  ^{n}\times\left\{  0,1\right\}  ^{m}$.
\ Suppose also that, if Alice sends a $q$-qubit quantum state $\left\vert
\psi_{x}\right\rangle $\ to Bob, then Bob can compute $f\left(  x,y\right)
$\ with bounded probability of error, for any $\left(  x,y\right)  \in S$.
\ Then given $\widetilde{O}\left(  qm^{4}\right)  $\ copies of $\left\vert
\psi_{x}\right\rangle $, Bob can compute $f\left(  x,y\right)  $\ for
\textit{every} $y$ such that $\left(  x,y\right)  \in S$%
\ simultaneously---again, even though Bob doesn't know which $y$'s satisfy
$\left(  x,y\right)  \in S$ (and therefore, which ones might be
\textquotedblleft dangerous\textquotedblright\ to measure).

\section{Preliminaries\label{PRELIM}}

In this section, we collect the (very basic) concepts and results of quantum
information that we'll need for this paper. \ In principle, no quantum
information background is needed to read the paper beyond this.

A \textit{mixed state} is the most general kind of state in quantum mechanics,
encompassing both superposition and ordinary probabilistic uncertainty. \ A
$D$-dimensional mixed state $\rho$ is described by a $D\times D$\ Hermitian
positive semidefinite matrix\ with $\operatorname{Tr}\left(  \rho\right)  =1$.
\ If $\rho$\ has rank $1$, then we call it a \textit{pure state}. \ At the
other extreme, if $\rho$\ is diagonal, then it simply describes a classical
probability distribution over $D$ outcomes, with $\Pr\left[  i\right]
=\rho_{ii}$. \ The state $\frac{I}{D}$, corresponding to the uniform
distribution over $D$ outcomes, is called the \textit{maximally mixed state}.

Given two mixed states $\rho$\ and $\sigma$, their \textit{trace distance} is
defined as%
\[
\left\Vert \rho-\sigma\right\Vert _{\operatorname*{tr}}:=\frac{1}{2}\sum
_{i}\left\vert \lambda_{i}\right\vert ,
\]
where the $\lambda_{i}$'s are the eigenvalues of $\rho-\sigma$. \ This is a
distance metric, which generalizes the variation distance between probability
distributions, and which equals the maximum bias with which $\rho$\ can be
distinguished from $\sigma$\ by a single-shot measurement.

Given a $D$-dimensional mixed state $\rho$, one thing we can do is to apply a
\textit{two-outcome measurement}, and see whether it accepts or rejects $\rho
$. \ Such a measurement---technically called a \textquotedblleft Positive
Operator Valued Measure\textquotedblright\ or \textquotedblleft
POVM\textquotedblright---can always be described by a $D\times D$\ Hermitian
matrix\ $E$\ with all eigenvalues in $\left[  0,1\right]  $ (so in particular,
$E$ is positive semidefinite). \ The measurement $E$ \textit{accepts} $\rho
$\ with probability $\operatorname{Tr}\left(  E\rho\right)  $, and
\textit{rejects} $\rho$\ with probability $1-\operatorname{Tr}\left(
E\rho\right)  $.

The POVM formalism doesn't tell us what happens to $\rho$\ after the
measurement, and indeed the post-measurement state could in general depend on
how $E$ is implemented. \ However, we have the following extremely useful
fact, which was called the \textquotedblleft Gentle Measurement
Lemma\textquotedblright\ by Winter \cite{winter:gentle}.

\begin{lemma}
[Gentle Measurement Lemma \cite{winter:gentle}]\label{gentle}Let $\rho$\ be a
mixed state, and let $E$ be a two-outcome measurement such that
$\operatorname{Tr}\left(  E\rho\right)  \geq1-\varepsilon$. \ Then after we
apply $E$ to $\rho$, assuming $E$ accepts, we can recover a post-measurement
state $\widetilde{\rho}$\ such that $\left\Vert \widetilde{\rho}%
-\rho\right\Vert _{\operatorname*{tr}}\leq2\sqrt{\varepsilon}$.
\end{lemma}

As a historical note, Aaronson \cite{aar:adv,aarbados}\ proved a variant of
Lemma \ref{gentle},\ which he called the \textquotedblleft Almost As Good As
New Lemma\textquotedblright;\ the main difference is that Aaronson's version
doesn't involve conditioning on the case that $E$\ accepts.

We'll also need a stronger fact, which goes back at least to Ambainis et
al.\ \cite{antv}, and which Aaronson \cite{aar:qmaqpoly,aarbados}\ called the
\textquotedblleft Quantum Union Bound.\textquotedblright\ \ Here we state the
strongest version, due to Wilde \cite{wilde}, although a less strong version
(involving the bound $M\sqrt{\varepsilon}$\ rather than $\sqrt{M\varepsilon}$)
would also have worked fine for us.

\begin{lemma}
[Quantum Union Bound \cite{antv,aar:qmaqpoly,wilde}]\label{qub}Let $\rho$ be a
mixed state, and let $E_{1},\ldots,E_{M}$\ be two-outcome measurements such
that $\operatorname{Tr}\left(  E_{i}\rho\right)  \geq1-\varepsilon$\ for all
$i\in\left[  M\right]  $. \ Then if $E_{1},\ldots,E_{M}$ are applied to $\rho$
in succession, the probability that they all accept is at least $1-2M\sqrt
{\varepsilon}$, and conditioned on all of them accepting, we can recover a
post-measurement state $\widetilde{\rho}$\ such that%
\[
\left\Vert \widetilde{\rho}-\rho\right\Vert _{\operatorname*{tr}}=O\left(
\sqrt{M\varepsilon}\right)  .
\]

\end{lemma}

\section{Gentle Search Procedure\label{SEARCH}}

We now develop a procedure that takes as input descriptions of two-outcome
measurements $E_{1},\ldots,E_{M}$, as well as $\operatorname*{polylog}%
M$\ copies of an unknown state $\rho$, and that searches for a measurement
$E_{i}$\ that accepts $\rho$\ with high probability. \ In Section
\ref{MAINSEC}, we'll then use this procedure as a key subroutine for solving
the shadow tomography problem.

Our starting point is a recent result of Harrow, Lin, and Montanaro
\cite{hlm}\ (their Corollary 11), which we state below for convenience.

\begin{theorem}
[Harrow, Lin, and Montanaro \cite{hlm}]\label{hlmthm}Let $\rho$\ be an unknown
mixed state,\ and let $E_{1},\ldots,E_{M}$\ be known two-outcome measurements.
\ Suppose we're promised that either

\begin{enumerate}
\item[(i)] there exists an $i\in\left[  M\right]  $\ such that
$\operatorname{Tr}\left(  E_{i}\rho\right)  \geq1-\epsilon$, or else

\item[(ii)] $\operatorname{Tr}\left(  E_{1}\rho\right)  +\cdots
+\operatorname{Tr}\left(  E_{M}\rho\right)  \leq\Delta M$.
\end{enumerate}

There is a test that uses one copy of $\rho$, and that accepts with
probability at least $\left(  1-\epsilon\right)  ^{2}/7$\ in case (i)\ and
with probability at most $4\Delta M$\ in case (ii).
\end{theorem}

Aaronson \cite{aar:qmaqpoly} had previously claimed a version of Theorem
\ref{hlmthm}, which he called the Quantum OR Bound. \ However, Aaronson's
proof had a mistake, which Harrow et al.\ \cite{hlm} both identified and fixed.

Briefly, Harrow et al.\ \cite{hlm} give two ways to prove Theorem
\ref{hlmthm}. \ The first way is by adapting the in-place amplification
procedure of Marriott and Watrous \cite{mw}. \ The second way is by preparing
a control qubit in the state $\frac{\left\vert 0\right\rangle +\left\vert
1\right\rangle }{\sqrt{2}}$, and then repeatedly applying $E_{i}$'s to $\rho
$\ conditional on the control qubit being $\left\vert 1\right\rangle $, while
also periodically measuring the control qubit in the $\left\{  \frac
{\left\vert 0\right\rangle +\left\vert 1\right\rangle }{\sqrt{2}}%
,\frac{\left\vert 0\right\rangle -\left\vert 1\right\rangle }{\sqrt{2}%
}\right\}  $\ basis to see whether applying the $E_{i}$'s has decohered the
control qubit. \ Harrow et al.\ show that, in case (i), \textit{either} some
$E_{i}$\ is likely to accept or else the control qubit is likely to be
decohered. \ In case (ii), on the other hand, one can upper-bound the
probability that either of these events happen using the Quantum Union Bound
(Lemma \ref{qub}).

Both of Harrow et al.'s procedures perform measurements on $\rho$ that involve
an ancilla register, and that are somewhat more complicated than the $E_{i}$'s
themselves. \ By contrast, the original procedure of Aaronson
\cite{aar:qmaqpoly} just applied the $E_{i}$'s in a random order. \ It remains
an open question whether the simpler procedure is sound.

In any case, by combining Theorem \ref{hlmthm} with a small amount of
amplification, we can obtain a variant of Theorem \ref{hlmthm} that's more
directly useful for us, and which we'll call \textquotedblleft
the\textquotedblright\ Quantum OR Bound in this paper.

\begin{lemma}
[Quantum OR Bound]\label{orbound}Let $\rho$\ be an unknown mixed state,\ and
let $E_{1},\ldots,E_{M}$\ be known two-outcome measurements. \ Suppose we're
promised that either

\begin{enumerate}
\item[(i)] there exists an $i\in\left[  M\right]  $\ such that
$\operatorname{Tr}\left(  E_{i}\rho\right)  \geq c$, or else

\item[(ii)] $\operatorname{Tr}\left(  E_{i}\rho\right)  \leq c-\varepsilon$
for all $i\in\left[  M\right]  $.
\end{enumerate}

We can distinguish these cases, with success probability at least $1-\delta
$,\ given $\rho^{\otimes k}$\ where $k=O\left(  \frac{\log1/\delta
}{\varepsilon^{2}}\log M\right)  $.
\end{lemma}

\begin{proof}
This essentially follows by combining Theorem \ref{hlmthm}\ with the Chernoff
bound. \ Assume without loss of generality that $M\geq50$, and let
$\ell=C\frac{\log M}{\varepsilon^{2}}$ for some sufficiently large constant
$C$. \ Also, let $E_{i}^{\ast}$\ be an amplified measurement that applies
$E_{i}$\ to each of $\ell$ registers, and that accepts if and only if the
number of accepting invocations is at least $\left(  c-\frac{\varepsilon}%
{2}\right)  \ell$. \ Then in case (i) we have%
\[
\operatorname{Tr}\left(  E_{i}^{\ast}\rho^{\otimes\ell}\right)  \geq1-\frac
{1}{M^{2}}%
\]
for some $i\in\left[  M\right]  $, while in case (ii) we have%
\[
\operatorname{Tr}\left(  E_{i}^{\ast}\rho^{\otimes\ell}\right)  \leq\frac
{1}{M^{2}}%
\]
for all $i$. \ So if we apply the procedure of Theorem \ref{hlmthm} to
$\rho^{\otimes\ell}$, then it accepts with probability at least (say)
$\frac{1}{8}$\ in case (i), or with probability at most $\frac{4}{M}$\ in case (ii).

We now just need $O\left(  \log1/\delta\right)  $\ rounds of further
amplification---involving a fresh copy of $\rho^{\otimes\ell}$\ in each
round---to push these acceptance probabilities to $1-\delta$\ or $\delta$\ respectively.
\end{proof}

Note that the procedure of Lemma \ref{orbound} requires performing collective
measurements on $O\left(  \frac{\log M}{\varepsilon^{2}}\right)  $\ copies of
$\rho$. \ On the positive side, though, the number of copies has no dependence
whatsoever on the Hilbert space dimension $D$.

Building on Lemma \ref{orbound}, we next want to give a \textit{search}
procedure: that is, a procedure that actually \textit{finds} a measurement
$E_{i}$ in our list that accepts $\rho$\ with high probability (if there is
one), rather than merely telling us whether such an $E_{i}$ exists. \ To do
this, we'll use the classic trick in computer science for reducing search
problems to decision problems:\ namely, binary search over the list
$E_{1},\ldots,E_{M}$.

The subtlety is that, as we run binary search, our lower bound on the
acceptance probability of the measurement $E_{i}$\ that we're isolating
degrades at each level of the recursion, while the error probability builds
up. \ Also, we need fresh copies of $\rho$\ at each level of the recursion.
\ Handling these issues will yield a procedure that uses roughly $\frac
{\log^{4}M}{\varepsilon^{2}}$\ copies of $\rho$,\ which we suspect is not tight.

\begin{lemma}
[Gentle Search]\label{searchlem}Let $\rho$\ be an unknown mixed state,\ and
let $E_{1},\ldots,E_{M}$\ be known two-outcome measurements. \ Suppose there
exists an $i\in\left[  M\right]  $\ such that $\operatorname{Tr}\left(
E_{i}\rho\right)  \geq c$. \ Then we can find a $j\in\left[  M\right]  $\ such
that $\operatorname{Tr}\left(  E_{j}\rho\right)  \geq c-\varepsilon$, with
success probability at least $1-\delta$, given $\rho^{\otimes k}$\ where%
\[
k=O\left(  \frac{\log^{4}M}{\varepsilon^{2}}\left(  \log\log M+\log\frac
{1}{\delta}\right)  \right)  .
\]

\end{lemma}

\begin{proof}
Assume without loss of generality that $M$ is a power of $2$. \ We will apply
Lemma \ref{orbound} recursively, using binary search to zero in on a $j$\ such
that $\operatorname{Tr}\left(  E_{j}\rho\right)  \geq c-\varepsilon$.

Divide the measurements into two sets, $S_{1}=\left\{  E_{1},\ldots
,E_{M/2}\right\}  $\ and $S_{2}=\left\{  E_{M/2+1},\ldots,E_{M}\right\}  $.
\ Also, let $\alpha:=\frac{\varepsilon}{\log_{2}M}$, and let $\beta
:=\frac{\delta}{\log_{2}M}$. \ Then as a first step, we call the subroutine
from Lemma \ref{orbound}\ to check, with success probability at least
$1-\beta$, whether

\begin{enumerate}
\item[(i)] there exists an\ $E\in S_{1}$\ such that $\operatorname{Tr}\left(
E\rho\right)  \geq c$ or

\item[(ii)] $\operatorname{Tr}\left(  E\rho\right)  \leq c-\alpha$\ for all
$E\in S_{1}$,
\end{enumerate}

\noindent promised that one of these is the case.

Note that the promise could be violated---but this simply means that, if the
subroutine returns (i), then we can assume only that there exists an $E\in
S_{1}$\ such that $\operatorname{Tr}\left(  E\rho\right)  \geq c-\alpha$.

Thus, if the subroutine returns (i), then we recurse on $S_{1}$. \ That is, we
divide $S_{1}$ into two sets both of size $\frac{M}{4}$, and then use Lemma
\ref{orbound} to find (again with success probability at least $1-\beta$) a
set that contains an $E$\ such that $\operatorname{Tr}\left(  E\rho\right)
\geq c-2\alpha$, assuming now that one of the two sets contains an $E$ such
that $\operatorname{Tr}\left(  E\rho\right)  \geq c-\alpha$. \ If the
subroutine returns (ii), then we do the same but with $S_{2}$.

We continue recursing in this way, identifying a set of size $\frac{M}{8}%
$\ that contains an $E$ such that $\operatorname{Tr}\left(  E\rho\right)  \geq
c-3\alpha$, then a set of size $\frac{M}{16}$\ that contains an $E$ such that
$\operatorname{Tr}\left(  E\rho\right)  \geq c-4\alpha$,\ and so on, until we
reach a singleton set. \ This gives us our index $j$ such that%
\begin{align*}
\operatorname{Tr}\left(  E_{j}\rho\right)   &  \geq c-\alpha\log_{2}M\\
&  =c-\varepsilon.
\end{align*}

By the union bound, together with the promise that there exists an
$i\in\left[  M\right]  $\ such that $\operatorname{Tr}\left(  E_{i}%
\rho\right)  \geq c$, the whole procedure succeeds with probability at least
$1-\beta\log_{2}M=1-\delta$. \ Meanwhile, within each of the $\log_{2}%
M$\ iterations of this procedure, Lemma \ref{orbound}\ tells us that the
number of copies of $\rho$\ that we need is%
\begin{align*}
O\left(  \frac{\log1/\beta}{\alpha^{2}}\log M\right)   &  =O\left(  \frac
{\log\frac{\log M}{\delta}}{\left(  \varepsilon/\log M\right)  ^{2}}\log
M\right) \\
&  =O\left(  \frac{\log^{3}M}{\varepsilon^{2}}\left(  \log\log M+\log\frac
{1}{\delta}\right)  \right)  .
\end{align*}
Therefore the total number of copies needed is%
\[
O\left(  \frac{\log^{4}M}{\varepsilon^{2}}\left(  \log\log M+\log\frac
{1}{\delta}\right)  \right)  .
\]

\end{proof}

\section{Main Result\label{MAINSEC}}

We're now ready to prove Theorem \ref{main}, which we restate for convenience.
\ Given an unknown $D$-dimensional mixed state $\rho$, and known two-outcome
measurements\ $E_{1},\ldots,E_{M}$, we'll show how to approximate
$\operatorname{Tr}\left(  E_{i}\rho\right)  $\ to within additive error
$\pm\varepsilon$, with success probability at least $1-\delta$, given
$\rho^{\otimes k}$\ where%
\[
k=\widetilde{O}\left(  \frac{\log1/\delta}{\varepsilon^{5}}\log^{4}M\cdot\log
D\right)  .
\]
Afterwards we'll remark on how to improve the $\varepsilon^{-5}$\ to
$\varepsilon^{-4}$, using a recent algorithm of Aaronson et al. \cite{achkn}
for online learning of quantum states.

\begin{proof}
[Proof of Theorem \ref{main}]At a high level, we'll use an iterative procedure
similar to the multiplicative weights update method: one that, at the $t^{th}%
$\ iteration, maintains a current hypothesis $\rho_{t}$\ about $\rho$. \ Since
we're not concerned here with computation time, we can assume that the entire
$D\times D$ density matrix of $\rho_{t}$\ is stored to suitable precision in a
classical memory, so that we can perform updates (in particular, involving
postselection) that wouldn't be possible were $\rho_{t}$\ an actual physical state.

Our initial hypothesis is that $\rho$\ is just the maximally mixed state,
$\rho_{0}=\frac{I}{D}$. \ Of course this hypothesis is unlikely to give
adequate predictions---but by using Lemma \ref{searchlem} as a subroutine,
we'll repeatedly refine the current hypothesis, $\rho_{t}$, to a
\textquotedblleft better\textquotedblright\ hypothesis $\rho_{t+1}$. \ The
procedure will terminate when we reach a hypothesis $\rho_{T}$\ such that%
\[
\left\vert \operatorname{Tr}\left(  E_{i}\rho_{T}\right)  -\operatorname{Tr}%
\left(  E_{i}\rho\right)  \right\vert \leq\varepsilon
\]
for all $i\in\left[  M\right]  $. \ For at that point, for each $i$, we can
just output $\operatorname{Tr}\left(  E_{i}\rho_{T}\right)  $\ as our additive
estimate for $\operatorname{Tr}\left(  E_{i}\rho\right)  $. \ 

At each iteration $t$, we'll use \textquotedblleft fresh\textquotedblright%
\ copies of $\rho$, in the course of refining $\rho_{t}$ to $\rho_{t+1}$.
\ Thus, we'll need to upper-bound both the total number $T$\ of iterations
until termination, \textit{and} the number of copies of $\rho$\ used in a
given iteration.

A key technical ingredient will be amplification. \ Let%
\[
q:=\frac{C}{\varepsilon^{2}}\left(  \log\log D+\log\frac{1}{\varepsilon
}\right)
\]
for some suitable constant $C$, and let $\rho^{\ast}:=\rho^{\otimes q}$.
\ Then, strictly speaking, our procedure will maintain a hypothesis $\rho
_{t}^{\ast}$\ about $\rho^{\ast}$: the initial hypothesis is the maximally
mixed state $\rho_{0}^{\ast}=\frac{I}{D^{q}}$; then we'll refine the
hypothesis to $\rho_{1}^{\ast}$, $\rho_{2}^{\ast}$, and so on. \ At any point,
we let $\rho_{t}$\ be the $D$-dimensional state obtained by choosing a
register of $\rho_{t}^{\ast}$\ uniformly at random, and tracing out the
remaining $q-1$\ registers.

Given the hypothesis $\rho_{t}$, for each $i\in\left[  M\right]  $, let
$E_{i,t,+}^{\ast}$\ be a two-outcome measurement on $\rho^{\otimes q}$\ that
applies $E_{i}$\ to each of the $q$ registers, and that accepts if and only if
the number of accepting invocations is at least $\left(  \operatorname{Tr}%
\left(  E_{i}\rho_{t}\right)  +\frac{3\varepsilon}{4}\right)  q$. \ Likewise,
let the measurement $E_{i,t,-}^{\ast}$\ apply $E_{i}$\ to each of the $q$
registers, and accept if and only if the number of accepting invocations is at
most $\left(  \operatorname{Tr}\left(  E_{i}\rho_{t}\right)  -\frac
{3\varepsilon}{4}\right)  q$.

Suppose $\operatorname{Tr}\left(  E_{i}\rho\right)  \geq\operatorname{Tr}%
\left(  E_{i}\rho_{t}\right)  +\varepsilon$. \ Then by a Chernoff bound, we
certainly have $\operatorname{Tr}\left(  E_{i,t,+}^{\ast}\rho^{\ast}\right)
\geq\frac{5}{6}$, provided the constant $C$ was sufficiently large: indeed,
for this we need only that $q$\ grows at least like $\frac{C}{\varepsilon^{2}%
}$. \ Likewise, if $\operatorname{Tr}\left(  E_{i}\rho\right)  \leq
\operatorname{Tr}\left(  E_{i}\rho_{t}\right)  -\varepsilon$, then
$\operatorname{Tr}\left(  E_{i,t,-}^{\ast}\rho^{\ast}\right)  \geq\frac{5}{6}$.

On the other hand, suppose $\left\vert \operatorname{Tr}\left(  E_{i}%
\rho\right)  -\operatorname{Tr}\left(  E_{i}\rho_{t}\right)  \right\vert
\leq\frac{\varepsilon}{2}$. \ Then again by a Chernoff bound, we have
$\operatorname{Tr}\left(  E_{i,t,+}^{\ast}\rho^{\ast}\right)  \leq\frac{1}{3}%
$\ and $\operatorname{Tr}\left(  E_{i,t,-}^{\ast}\rho^{\ast}\right)  \leq
\frac{1}{3}$, provided the constant $C$ is sufficiently large.

We can now give the procedure to update the hypothesis $\rho_{t}^{\ast}$\ to
$\rho_{t+1}^{\ast}$. \ Let $\beta:=\frac{\delta\varepsilon^{4}}{\log^{2}D}$.
\ Then at each iteration $t\geq0$, we do the following:

\begin{itemize}
\item Use Lemma \ref{searchlem} to search for an index $j\in\left[  M\right]
$\ such that $\operatorname{Tr}\left(  E_{j,t,+}^{\ast}\rho^{\ast}\right)
\geq\frac{2}{3}$, promised that there exists such a $j$ with
$\operatorname{Tr}\left(  E_{j,t,+}^{\ast}\rho^{\ast}\right)  \geq\frac{5}{6}%
$\ (which we call the $+$\ case); or for a $j\in\left[  M\right]  $\ such that
$\operatorname{Tr}\left(  E_{j,t,-}^{\ast}\rho^{\ast}\right)  \geq\frac{2}{3}%
$, promised that there exists a $j$\ such that $\operatorname{Tr}\left(
E_{j,t,-}^{\ast}\rho^{\ast}\right)  \geq\frac{5}{6}$ (which we call the
$-$\ case). \ Set the parameters so that, assuming that one or both promises
hold, the search succeeds with probability at least $1-\beta$.

\item If no $j$ is found that satisfies either condition, then halt and return
$\rho_{t}$\ as the hypothesis state: in other words, return $\operatorname{Tr}%
\left(  E_{i}\rho_{t}\right)  $\ as the estimate for $\operatorname{Tr}\left(
E_{i}\rho\right)  $, for all $i\in\left[  M\right]  $.

\item Otherwise, if a $j$ \textit{is} found satisfying one of the conditions,
then let $F_{t}$\ be a measurement that applies $E_{j}$\ to each of the $q$
registers, and that accepts if and only if the number of accepting invocations
is at least $\left(  \operatorname{Tr}\left(  E_{j}\rho_{t}\right)
+\frac{\varepsilon}{4}\right)  q$ (in the $+$ case), or at most $\left(
\operatorname{Tr}\left(  E_{j}\rho_{t}\right)  -\frac{\varepsilon}{4}\right)
q$ (in the $-$ case).

\item Let $\rho_{t+1}^{\ast}$\ be the state obtained by starting from
$\rho_{t}^{\ast}$, and then postselecting on $F_{t}$\ accepting.
\end{itemize}

Our central task is to prove an upper bound, $T$, on the number of iterations
of the above procedure until it terminates with states $\rho_{T}^{\ast}$\ and
$\rho_{T}$\ such that $\left\vert \operatorname{Tr}\left(  E_{i}\rho
_{T}\right)  -\operatorname{Tr}\left(  E_{i}\rho\right)  \right\vert
\leq\varepsilon$\ for all $i\in\left[  M\right]  $.

Assume, in what follows, that every invocation of Lemma \ref{searchlem}
succeeds in finding a $j$ such that $\operatorname{Tr}\left(  E_{j,t,+}^{\ast
}\rho^{\ast}\right)  \geq\frac{2}{3}$\ or $\operatorname{Tr}\left(
E_{j,t,-}^{\ast}\rho^{\ast}\right)  \geq\frac{2}{3}$. \ Later we will
lower-bound the probability that this indeed happens.

To upper-bound $T$,\ let%
\[
p_{t}=\operatorname{Tr}\left(  F_{0}\rho_{0}^{\ast}\right)  \cdot\cdots
\cdot\operatorname{Tr}\left(  F_{t-1}\rho_{t-1}^{\ast}\right)
\]
be the probability that the first $t$ postselection steps all succeed.

Then, on the one hand, we claim that $p_{t+1}\leq\left(  1-\Omega\left(
\varepsilon\right)  \right)  p_{t}$ for all $t$. \ To see this, note that when
we run $F_{t}$ on the state $\rho_{t}^{\ast}$, the expected number of
invocations of $E_{j}$\ that accept is exactly $\operatorname{Tr}\left(
E_{j}\rho_{t}\right)  q$. \ We can't treat these invocations as independent
events, because the state $\rho_{t}^{\ast}$\ could be correlated or entangled
across its $q$\ registers in some unknown way. \ Regardless of correlations,
though, by Markov's inequality, in the $+$\ case we have%
\[
\frac{p_{t+1}}{p_{t}}=\operatorname{Tr}\left(  F_{t}\rho_{t}^{\ast}\right)
\leq\frac{\operatorname{Tr}\left(  E_{j}\rho_{t}\right)  q}{\left(
\operatorname{Tr}\left(  E_{j}\rho_{t}\right)  +\frac{\varepsilon}{4}\right)
q}=1-\Omega\left(  \varepsilon\right)  .
\]
Similarly, in the $-$\ case we have%
\[
\frac{p_{t+1}}{p_{t}}=\operatorname{Tr}\left(  F_{t}\rho_{t}^{\ast}\right)
\leq\frac{\left(  1-\operatorname{Tr}\left(  E_{j}\rho_{t}\right)  \right)
q}{\left(  1-\operatorname{Tr}\left(  E_{j}\rho_{t}\right)  +\frac
{\varepsilon}{4}\right)  q}=1-\Omega\left(  \varepsilon\right)  .
\]
Thus $p_{t}\leq\left(  1-\varepsilon\right)  ^{\Omega\left(  t\right)  }$.

On the other hand, we also claim that $p_{t}\geq\frac{0.9}{D^{q}}$ for all
$t=o\left(  \frac{\log^{2}D}{\varepsilon^{4}}\right)  $. \ To see this:
suppose that at iteration $t$, we had used $\rho^{\ast}$ rather than $\rho
_{t}^{\ast}$ as the hypothesis state---except still choosing the index
$j\in\left[  M\right]  $\ as if the hypothesis was $\rho_{t}^{\ast}$. \ In
that case, when we applied $F_{t}$\ to $\rho^{\ast}$, the expected number of
accepting invocations of $E_{j}$\ would be exactly $\operatorname{Tr}\left(
E_{j}\rho\right)  q$.

Consider for concreteness the $+$\ case; the $-$\ case\ is precisely
analogous. \ By the assumption that the search for $j$ succeeded, we have
$\operatorname{Tr}\left(  E_{j,t,+}^{\ast}\rho^{\ast}\right)  \geq\frac{2}{3}%
$. \ In other words: when we apply $E_{j}$\ to $q$ copies of $\rho$, the
number of invocations that accept is at least $\left(  \operatorname{Tr}%
\left(  E_{i}\rho_{t}\right)  +\frac{3\varepsilon}{4}\right)  q$, with
probability at least $\frac{2}{3}$. \ Recall that $q\geq\frac{C}%
{\varepsilon^{2}}$\ for some sufficiently large constant $C$. \ So since the
$q$\ copies\ of $\rho$\ really \textit{are} independent, in this case we can
use a Chernoff bound to conclude that $\operatorname{Tr}\left(  E_{i}%
\rho\right)  >\operatorname{Tr}\left(  E_{i}\rho_{t}\right)  +\frac
{\varepsilon}{2}$.

Now we consider $1-\operatorname{Tr}\left(  F_{t}\rho^{\ast}\right)  $: that
is, the probability that $F_{t}$ rejects $\rho^{\ast}$. \ This is just the
probability that, when we apply $E_{j}$\ to $q$ copies of $\rho$, the number
of invocations that accept is less than $\left(  \operatorname{Tr}\left(
E_{j}\rho_{t}\right)  +\frac{\varepsilon}{4}\right)  q$. \ Since the copies of
$\rho$\ are independent, and since (by the above) the expected number of
accepting invocations is at least $\left(  \operatorname{Tr}\left(  E_{i}%
\rho_{t}\right)  +\frac{\varepsilon}{2}\right)  q$, we can again use a
Chernoff bound to conclude that%
\begin{align*}
1-\operatorname{Tr}\left(  F_{t}\rho^{\ast}\right)   &  \leq\exp\left(
-\Omega\left(  \varepsilon^{2}q\right)  \right) \\
&  \leq\exp\left(  -\Omega\left(  \varepsilon^{2}\cdot\frac{C}{\varepsilon
^{2}}\left(  \log\log D+\log\frac{1}{\varepsilon}\right)  \right)  \right) \\
&  \leq\frac{\varepsilon^{4}}{\log^{2}D},
\end{align*}
provided we took the constant $C$ sufficiently large.

By Lemma \ref{qub} (the Quantum Union Bound), this means that, even if we
applied $F_{0},\ldots,F_{T-1}$\ to $\rho^{\ast}$\ in succession, the
probability that \textit{any} of them would reject is at most%
\[
O\left(  \sqrt{\frac{T\varepsilon^{4}}{\log^{2}D}}\right)  =O\left(
\frac{\sqrt{T}\varepsilon^{2}}{\log D}\right)  .
\]
Hence, so as long as $T=o\left(  \frac{\log^{2}D}{\varepsilon^{4}}\right)  $,
all of the $F_{t}$'s accept $\rho^{\ast}$ with probability at least (say)
$0.9$.

But we can always decompose the maximally mixed state, $\rho_{0}^{\ast}%
=\frac{I}{D^{q}}$, as%
\[
\rho_{0}^{\ast}=\frac{1}{D^{q}}\rho^{\ast}+\left(  1-\frac{1}{D^{q}}\right)
\sigma,
\]
where $\sigma$\ is some mixed state. \ This means that, in the
\textquotedblleft real\textquotedblright\ situation (i.e., when we run the
procedure with initial state $\rho_{0}^{\ast}$), all of the $F_{t}$'s accept
$\rho_{0}^{\ast}$ with probability at least $\frac{0.9}{D^{q}}$, as claimed.

Combining the two claims above---namely, $p_{t}\leq\left(  1-\varepsilon
\right)  ^{\Omega\left(  t\right)  }$\ and $p_{t}\geq\frac{0.9}{D^{q}}$---we
get%
\[
\left(  1-\varepsilon\right)  ^{\Omega\left(  t\right)  }\geq\frac{0.9}{D^{q}%
}.
\]
Solving for $t$ now yields%
\begin{align*}
t  &  =O\left(  \frac{q\log D}{\varepsilon}\right) \\
&  =O\left(  \frac{\log D}{\varepsilon}\cdot\frac{C}{\varepsilon^{2}}\left(
\log\log D+\log\frac{1}{\varepsilon}\right)  \right)  .
\end{align*}
This then gives us the desired upper bound on $T$, and justifies the
assumption we made before that $T=o\left(  \frac{\log^{2}D}{\varepsilon^{4}%
}\right)  $.

Meanwhile, by Lemma \ref{searchlem} together with the union bound, the
probability that all $T$ invocations of Lemma \ref{searchlem} succeed at
finding a suitable index $j$\ is at least%
\[
1-T\beta=1-o\left(  \frac{\log^{2}D}{\varepsilon^{4}}\cdot\frac{\delta
\varepsilon^{4}}{\log^{2}D}\right)  \geq1-\delta,
\]
as needed.

Finally, we upper-bound the total number of copies of $\rho$\ used by the
procedure. \ Within each iteration $t$, the bound of Lemma \ref{searchlem}
tells us that we need%
\[
\ell=O\left(  \log^{4}M\left(  \log\log M+\log\frac{1}{\beta}\right)  \right)
=O\left(  \log^{4}M\left(  \log\log M+\log\log D+\log\frac{1}{\varepsilon
}+\log\frac{1}{\delta}\right)  \right)
\]
copies of the amplified state $\rho^{\ast}$. \ Since $\rho^{\ast}%
=\rho^{\otimes q}$, this translates to $q\ell$\ copies of $\rho$\ itself in
each of the $T$ iterations, or\
\begin{align*}
Tq\ell &  =O\left(  \frac{q\log D}{\varepsilon}\cdot q\ell\right) \\
&  =O\left(  \frac{\log D}{\varepsilon}\cdot q^{2}\cdot\ell\right) \\
&  =O\left(  \frac{\log D}{\varepsilon}\cdot\left(  \frac{\log\log D+\log
\frac{1}{\varepsilon}}{\varepsilon^{2}}\right)  ^{2}\cdot\log^{4}M\left(
\log\log M+\log\log D+\log\frac{1}{\varepsilon}+\log\frac{1}{\delta}\right)
\right) \\
&  =\widetilde{O}\left(  \frac{\log1/\delta}{\varepsilon^{5}}\cdot\log
^{4}M\cdot\log D\right)
\end{align*}
copies of $\rho$\ total.
\end{proof}

Let us briefly indicate how recent work by Aaronson et al. \cite{achkn} can be
used to improve the dependence on $\varepsilon$\ in Theorem \ref{main}\ from
$1/\varepsilon^{5}$\ to $1/\varepsilon^{4}$. \ Apart from some low-order terms
needed to amplify success probabilities, the proof of Theorem \ref{main} can
be seen as simply a composition of two algorithms:

\begin{enumerate}
\item[(1)] an algorithm for learning a $D$-dimensional mixed state $\rho$,
starting with the initial hypothesis $\rho_{0}=\frac{I}{D}$ and then
repeatedly refining it, given a sequence of two-outcome measurements on which
the current hypothesis $\rho_{t}$\ is wrong by more then $\varepsilon$,

\item[(2)] an algorithm (namely, the Gentle Search Procedure of Lemma
\ref{searchlem}) for actually \textit{finding} the measurements on which the
current hypothesis is wrong by more than $\varepsilon$, in the context of
shadow tomography.
\end{enumerate}

Our algorithm (1) uses $\widetilde{O}\left(  \frac{\log D}{\varepsilon^{3}%
}\right)  $ refinement iterations, while algorithm (2)---which gets invoked
inside every iteration---uses $\widetilde{O}\left(  \frac{\log^{4}%
M}{\varepsilon^{2}}\right)  $\ copies of $\rho$. \ Multiplying these two
bounds is what produces our final sample complexity of%
\[
\widetilde{O}\left(  \frac{\log^{4}M\cdot\log D}{\varepsilon^{5}}\right)  .
\]
Motivated by a slightly different setting (namely, online learning of quantum
states), Aaronson et al. \cite{achn}\ have now given a black-box way, using
matrix multiplicative weights and convex optimization, to improve the number
of iterations in (1) to $\widetilde{O}\left(  \frac{\log D}{\varepsilon^{2}%
}\right)  $, which matches an information-theoretic lower bound for (1) (see
\cite{aar:learn}). \ Using that in place of the postselection procedure of
Theorem \ref{main}, and combining with the Gentle Search Procedure, yields a
slightly improved bound of%
\[
\widetilde{O}\left(  \frac{\log^{4}M\cdot\log D}{\varepsilon^{4}}\right)
\]
on the sample complexity of shadow tomography.

\section{Lower Bound\label{LBAPPENDIX}}

We now show that, for purely information-theoretic reasons, any solution to
the shadow tomography problem requires at least $\Omega\left(  \frac
{\min\left\{  D^{2},\log M\right\}  }{\varepsilon^{2}}\right)  $\ copies of
$\rho$. \ Indeed, even the classical special case of the problem requires
$\Omega\left(  \frac{\min\left\{  D,\log M\right\}  }{\varepsilon^{2}}\right)
$\ copies. \ These are the best lower bounds for shadow tomography that we
currently know.

\subsection{Classical Special Case}

We start with the special case where we're given $k$ samples from an unknown
distribution $\mathcal{D}$\ over a $D$-element set, and our goal is to learn
$\operatorname{E}_{x\sim\mathcal{D}}\left[  f_{i}\left(  x\right)  \right]
$\ to within additive error $\pm\varepsilon$, for each of $M$ known Boolean
functions $f_{1},\ldots,f_{M}:\left[  D\right]  \rightarrow\left\{
0,1\right\}  $.

\begin{theorem}
\label{lower}Any strategy for shadow tomography---i.e., for estimating
$\operatorname{Tr}\left(  E_{i}\rho\right)  $\ to within additive error
$\varepsilon$\ for all $i\in\left[  M\right]  $, with success probability at
least (say) $2/3$---requires $\Omega\left(  \frac{\min\left\{  D,\log
M\right\}  }{\varepsilon^{2}}\right)  $\ copies of the $D$-dimensional mixed
state $\rho$. \ Furthermore, this is true even for the classical special case
(i.e., where $\rho$\ and the $E_{i}$'s are all diagonal).
\end{theorem}

\begin{proof}
Set $N:=\left\lfloor \min\left\{  D,\log_{2}M\right\}  \right\rfloor $. \ Our
distributions will be over $N$-element sets. \ Also, for some constant
$c\in\left(  1,2\right)  $, set $K:=\left\lfloor c^{N}\right\rfloor $, so that
$K\leq M$. \ We will have $K$ two-outcome measurements.

The first step is to choose $K$ subsets $S_{1},\ldots,S_{K}\subset\left[
N\right]  $ uniformly and independently, subject to the constraint that
$\left\vert S_{i}\right\vert =\frac{N}{2}$\ for each $i\in\left[  K\right]  $.
\ As long as $c$ is sufficiently small, by a Chernoff bound and union bound,
it's not hard to see that with probability $1-o\left(  1\right)  $\ over the
choice of $S_{i}$'s, we'll have%
\begin{equation}
\left\vert \left\vert S_{i}\cap S_{j}\right\vert -\frac{N}{4}\right\vert
\leq\frac{N}{12}\label{intersect}%
\end{equation}
for all $i\neq j$. \ So fix a choice of subsets $S_{i}$\ for which this happens.

Next, for each $i\in\left[  K\right]  $, let $\mathcal{D}_{i}$\ be the
probability distribution over $x\in\left[  N\right]  $ defined as follows:
choose $x$ uniformly from $S_{i}$ with probability $\frac{1}{2}+3\varepsilon$,
or uniformly from $\left[  N\right]  \setminus S_{i}$ with probability
$\frac{1}{2}-3\varepsilon$.\ \ Also, for each $i\in\left[  K\right]  $, let
$E_{i}$\ be the standard basis measurement that accepts each basis state
$x\in\left[  N\right]  $\ if $x\in S_{i}$ and rejects it otherwise. \ Then by
construction, for all $i\in\left[  K\right]  $, we have%
\[
\Pr_{x\sim\mathcal{D}_{i}}\left[  E_{i}\left(  x\right)  \text{ accepts}%
\right]  =\frac{1}{2}+3\varepsilon.
\]
Also, by (\ref{intersect}), for all $i\neq j$\ we have%
\[
\left\vert \Pr_{x\sim\mathcal{D}_{i}}\left[  E_{j}\left(  x\right)  \text{
accepts}\right]  -\frac{1}{2}\right\vert \leq\frac{\varepsilon}{2}.
\]
It follows that, if we can estimate $\Pr_{x\sim\mathcal{D}_{i}}\left[
E_{j}\left(  x\right)  \text{ accepts}\right]  $\ to within additive error
$\pm\varepsilon$\ for every $j\in\left[  K\right]  $, that is enough to
determine $i\in\left[  K\right]  $.

Notice that, if we choose $i\in\left[  K\right]  $ uniformly at random, then
it contains $\log_{2}\left(  K\right)  =\Omega\left(  N\right)  $\ bits of
information. \ Thus, let $X=\left(  x_{1},\ldots,x_{T}\right)  $\ consist of
$T$ independent samples from $\mathcal{D}_{i}$. \ Then in order for it to be
information-theoretically \textit{possible} to learn $i$ from $X$, the mutual
information $\operatorname{I}\left(  X;i\right)  $\ must be at least $\log
_{2}\left(  K\right)  $. \ We can now write%
\begin{align*}
\operatorname{I}\left(  X;i\right)    & =\operatorname*{H}\left(  X\right)
-\operatorname*{H}\left(  X|i\right)  \\
& =\sum_{t=1}^{T}\left(  \operatorname*{H}\left(  x_{t}~|~x_{1},\ldots
,x_{t-1}\right)  -\operatorname*{H}\left(  x_{t}~|~i,x_{1},\ldots
,x_{t-1}\right)  \right)  \\
& =\sum_{t=1}^{T}\left(  \operatorname*{H}\left(  x_{t}~|~x_{1},\ldots
,x_{t-1}\right)  -\operatorname*{H}\left(  x_{t}~|~i\right)  \right)  \\
& \leq\sum_{t=1}^{T}\left(  \operatorname*{H}\left(  x_{t}\right)
-\operatorname*{H}\left(  x_{t}~|~i\right)  \right)  \\
& \leq\sum_{t=1}^{T}\left(  \log_{2}N-\operatorname*{H}\left(  \mathcal{D}%
_{i}\right)  \right)  .
\end{align*}
Here the third line follows since $x_{t}$\ has no further dependence on
$x_{1},\ldots,x_{t-1}$\ once we've already conditioned on $i$, and the last
since a distribution over $\left[  N\right]  $\ can have entropy at most
$\log_{2}N$.

Now for all $i\in\left[  K\right]  $,%
\begin{align*}
\operatorname*{H}\left(  \mathcal{D}_{i}\right)    & =\sum_{x=1}^{N}%
\Pr_{\mathcal{D}_{i}}\left[  x\right]  \log_{2}\frac{1}{\Pr_{\mathcal{D}_{i}%
}\left[  x\right]  }\\
& =\frac{N}{2}\left(  \frac{1/2+3\varepsilon}{N/2}\right)  \log_{2}\left(
\frac{N/2}{1/2+3\varepsilon}\right)  +\frac{N}{2}\left(  \frac
{1/2-3\varepsilon}{N/2}\right)  \log_{2}\left(  \frac{N/2}{1/2-3\varepsilon
}\right)  \\
& =\log_{2}N-\left[  1-\left(  \frac{1}{2}+3\varepsilon\right)  \log
_{2}\left(  \frac{1}{1/2+3\varepsilon}\right)  -\left(  \frac{1}%
{2}-3\varepsilon\right)  \log_{2}\left(  \frac{1}{1/2-3\varepsilon}\right)
\right]  \\
& \geq\log_{2}N-O\left(  \varepsilon^{2}\right)  .
\end{align*}
Combining,%
\[
\operatorname{I}\left(  X;i\right)  =O\left(  T\varepsilon^{2}\right)  .
\]
We conclude that, in order to achieve mutual information $\operatorname{I}%
\left(  X;i\right)  =\Omega\left(  N\right)  $, so that $X$\ can determine
$i$,%
\[
T=\Omega\left(  \frac{N}{\varepsilon^{2}}\right)  =\Omega\left(  \frac
{\min\left\{  D,\log M\right\}  }{\varepsilon^{2}}\right)
\]
samples from $\mathcal{D}_{i}$\ are information-theoretically necessary.
\end{proof}

Let's observe that, in the classical special case, Theorem \ref{lower} has a
matching upper bound.

\begin{proposition}
\label{classical}Let $\mathcal{D}$\ be an unknown distribution over $\left[
D\right]  $, and let $f_{1},\ldots,f_{M}:\left[  D\right]  \rightarrow\left\{
0,1\right\}  $\ be known Boolean functions. \ Then given%
\[
O\left(  \frac{1}{\varepsilon^{2}}\min\left\{  D\log\frac{1}{\delta},\log
\frac{M}{\delta}\right\}  \right)
\]
independent samples from $\mathcal{D}$, it's possible to estimate
$\operatorname{E}_{x\sim\mathcal{D}}\left[  f_{i}\left(  x\right)  \right]
$\ to within additive error $\pm\varepsilon$\ for each $i\in\left[  M\right]
$, with success probability at least $1-\delta$.
\end{proposition}

\begin{proof}
First, to achieve $O\left(  \frac{D\log1/\delta}{\varepsilon^{2}}\right)  $:
it's a folklore fact (see for example \cite{kops}) that, given an unknown
distribution $\mathcal{D}$\ over $\left[  D\right]  $, with high probability,
$O\left(  \frac{D}{\varepsilon^{2}}\right)  $\ samples suffice to learn
$\mathcal{D}$\ up to $\varepsilon$\ error in variation distance. \ This
probability can then be amplified to $1-\delta$\ at the cost of $O\left(
\log\frac{1}{\delta}\right)  $\ repetitions. \ And of course, once
$\mathcal{D}$\ has been so learned, we can then estimate any $\operatorname{E}%
_{x\sim\mathcal{D}}\left[  f_{i}\left(  x\right)  \right]  $\ we like to
within additive error $\pm\varepsilon$.

Second, to achieve $O\left(  \frac{\log M/\delta}{\varepsilon^{2}}\right)  $:
the strategy is simply, for each $i\in\left[  M\right]  $, to output the
empirical mean of $f_{i}\left(  x\right)  $\ on the observed samples. \ By a
Chernoff bound, this strategy fails for any particular $i$ with probability at
most%
\[
\exp\left(  -\varepsilon^{2}\frac{\log M/\delta}{\varepsilon^{2}}\right)
\leq\frac{\delta}{M},
\]
provided the constant in the big-$O$ is sufficiently large. \ So by the union
bound, it succeeds for every $i$ simultaneously with probability at least
$1-\delta$.
\end{proof}

We now derive some additional consequences from the proof of Theorem
\ref{lower}. \ Note that, because recovering the secret index $i$\ boils down
to identifying a \textit{single} measurement $E_{j}$\ such that%
\[
\Pr_{x\sim\mathcal{D}_{i}}\left[  E_{j}\left(  x\right)  \text{ accepts}%
\right]  >\frac{1}{2}+\varepsilon,
\]
the example\ from the proof also shows that the search problem of Lemma
\ref{searchlem}, considered by itself, already requires $\Omega\left(
\frac{\min\left\{  D,\log M\right\}  }{\varepsilon^{2}}\right)  $\ copies of
$\rho$---again, even in the classical special case, where $\rho$\ and the
$E_{j}$'s are diagonal.

Indeed, we now make a further observation: for the specific example in Theorem
\ref{lower}, going from Lemma \ref{orbound} (the Quantum OR Bound) to Lemma
\ref{searchlem} (the gentle search procedure) requires almost no blowup in the
number of copies of $\rho$. \ This is true for two reasons. \ First, $\rho$ is
classical, so in the proof of Lemma \ref{searchlem}, we can reuse the same
copies of $\rho$\ from one binary search iteration to the next. \ Second, the
example assumed a \textquotedblleft promise gap\textquotedblright---i.e., for
all $i,j$, the distribution $\mathcal{D}_{i}$ is either $3\varepsilon$-biased
toward $E_{j}$\ accepting or else\ not even $\varepsilon$-biased toward
it---so there is no need to replace $\varepsilon$\ by $\frac{\varepsilon}{\log
M}$. \ The only amplification we need is $\log\log M$\ repetitions, to push
the failure probability per binary search iteration down to $\frac{1}{\log M}$.

We conclude that the Quantum OR Bound---again, even in the classical special
case---requires $\Omega\left(  \frac{\min\left\{  D,\log M\right\}
}{\varepsilon^{2}\log\log M}\right)  $\ copies of $\rho$, since otherwise we
could solve the search problem on the Hadamard example with $o\left(
\frac{\min\left\{  D,\log M\right\}  }{\varepsilon^{2}}\right)  $\ samples,
contradicting our previous reasoning. \ In other words, Lemma \ref{orbound}%
\ is close to tight.

\subsection{General Case}

We now show that, when the goal is to learn $D$-dimensional quantum mixed
states, Theorem \ref{lower} can be tightened to $\Omega\left(  \frac
{\min\left\{  D^{2},\log M\right\}  }{\varepsilon^{2}}\right)  $. \ Note that
this subsumes the lower bounds of O'Donnell and Wright \cite{owright}\ and
Haah et al.\ \cite{hhjwy}, reducing to the latter as we let $M$ be arbitrarily large.

We'll need the following lemma, which is a special case of a more general
result proved by Hayden, Leung, and Winter \cite{hlw}.

\begin{lemma}
[{\cite[Lemma III.5]{hlw}}]\label{hlwlemma}Let $N$\ be even, and let $S$\ be a
subspace of $\mathbb{C}^{N}$\ chosen uniformly at random subject to
$\dim\left(  S\right)  =\frac{N}{2}$. \ Let $\mathbb{P}_{S}$\ be the
projection onto $S$, and $\rho_{S}=\frac{2}{N}\mathbb{P}_{S}$ be the maximally
mixed state projected onto $S$. \ Then for any fixed subspace $T\leq
\mathbb{C}^{N}$\ of dimension $\frac{N}{2}$---for example, the one spanned by
$\left\vert 1\right\rangle ,\ldots,\left\vert N/2\right\rangle $---we have%
\[
\Pr_{S}\left[  \left\vert \operatorname{Tr}\left(  \mathbb{P}_{T}\rho
_{S}\right)  -\frac{1}{4}\right\vert >\frac{1}{20}\right]  \leq e^{-\Omega
\left(  N^{2}\right)  }.
\]

\end{lemma}

We can now prove the lower bound.

\begin{theorem}
\label{lowerq}Any strategy for shadow tomography---i.e., for estimating
$\operatorname{Tr}\left(  E_{i}\rho\right)  $\ to within additive error
$\varepsilon$\ for all $i\in\left[  M\right]  $, with success probability at
least (say) $2/3$---requires $\Omega\left(  \frac{\min\left\{  D^{2},\log
M\right\}  }{\varepsilon^{2}}\right)  $\ copies of the $D$-dimensional mixed
state $\rho$.
\end{theorem}

\begin{proof}
Set $N:=\left\lfloor \min\left\{  D,\sqrt{\log_{2}M}\right\}  \right\rfloor $.
\ Our mixed states will be $N$-dimensional. \ Also, for some constant
$c\in\left(  1,2\right)  $, set $K:=\left\lfloor c^{N^{2}}\right\rfloor $, so
that $K\leq M$. \ We will have $K$ two-outcome measurements.

The first step is to choose $K$ subspaces $S_{1},\ldots,S_{K}\leq
\mathbb{C}^{N}$ Haar-randomly and independently, subject to the constraint
that $\dim\left(  S_{i}\right)  =\frac{N}{2}$\ for each $i\in\left[  K\right]
$. \ Let $\mathbb{P}_{i}$\ be the projection onto $S_{i}$, and let $\rho
_{i}:=\frac{2}{N}\mathbb{P}_{i}$ be the maximally mixed state projected onto
$S_{i}$. \ Then as long as $c$ is sufficiently small, Lemma \ref{hlwlemma}
tells us that with probability $1-o\left(  1\right)  $\ over the choice of
$S_{i}$'s, we'll have%
\begin{equation}
\left\vert \operatorname{Tr}\left(  \mathbb{P}_{i}\rho_{j}\right)  -\frac
{1}{2}\right\vert \leq\frac{1}{12}\label{intersectq}%
\end{equation}
for all $i\neq j$. \ So fix a choice of subspaces $S_{i}$\ for which this happens.

Next, for each $i\in\left[  K\right]  $, define the mixed state%
\[
\sigma_{i}:=\left(  1-6\varepsilon\right)  \frac{\mathbb{I}}{N}+6\varepsilon
\rho_{i}.
\]

Our two-outcome POVMs will simply be the projectors $\mathbb{P}_{1}%
,\ldots,\mathbb{P}_{K}$. \ Then by construction, for all $i\in\left[
K\right]  $, we have%
\begin{align*}
\operatorname{Tr}\left(  \mathbb{P}_{i}\sigma_{i}\right)    & =\left(
1-6\varepsilon\right)  \frac{\operatorname{Tr}\left(  \mathbb{P}_{i}\right)
}{N}+6\varepsilon\operatorname{Tr}\left(  \mathbb{P}_{i}\rho_{i}\right)  \\
& =\left(  1-6\varepsilon\right)  \frac{N/2}{N}+6\varepsilon\\
& =\frac{1}{2}+3\varepsilon.
\end{align*}
Also, by (\ref{intersectq}), for all $i\neq j$\ we have%
\begin{align*}
\left\vert \operatorname{Tr}\left(  \mathbb{P}_{j}\sigma_{i}\right)  -\frac
{1}{2}\right\vert  & =\left\vert \left(  \left(  1-6\varepsilon\right)
\frac{\operatorname{Tr}\left(  \mathbb{P}_{j}\right)  }{N}-\left(  \frac{1}%
{2}-3\varepsilon\right)  \right)  +\left(  6\varepsilon\operatorname{Tr}%
\left(  \mathbb{P}_{j}\rho_{i}\right)  -3\varepsilon\right)  \right\vert \\
& =6\varepsilon\left\vert \operatorname{Tr}\left(  \mathbb{P}_{j}\rho
_{i}\right)  -\frac{1}{2}\right\vert \\
& \leq\frac{\varepsilon}{2}.
\end{align*}
It follows that, if we can estimate $\operatorname{Tr}\left(  \mathbb{P}%
_{j}\sigma_{i}\right)  $\ to within additive error $\pm\varepsilon$\ for every
$j\in\left[  K\right]  $, that is enough to determine $i\in\left[  K\right]  $.

Notice that, if we choose $i\in\left[  K\right]  $ uniformly at random, then
it contains $\log_{2}\left(  K\right)  =\Omega\left(  N^{2}\right)  $\ bits of
information. \ Thus, let%
\[
\zeta:=\operatorname{E}_{i\in\left[  K\right]  }\left[  \sigma_{i}^{\otimes
T}\right]  .
\]
Then in order for it to be information-theoretically \textit{possible} to
learn $i$ from $\zeta$, the quantum mutual information $\operatorname{I}%
\left(  \zeta;i\right)  $\ must be at least $\log_{2}\left(  K\right)  $.
\ Since $i$ is classical, we can now write%
\begin{align*}
\operatorname{I}\left(  \zeta;i\right)    & =\operatorname{S}\left(
\zeta\right)  -\operatorname{S}\left(  \zeta~|~i\right)  \\
& =\operatorname{S}\left(  \zeta\right)  -\operatorname{S}\left(  \sigma
_{i}^{\otimes T}\right)  \\
& \leq T\left(  \log_{2}N-\operatorname{S}\left(  \sigma_{i}\right)  \right)
,
\end{align*}
where $\operatorname{S}$\ is von Neumann entropy, and the third line follows
because $\zeta$\ is an $N^{T}$-dimensional state.

Now for all $i\in\left[  K\right]  $, if we let $\lambda_{i,1},\ldots
,\lambda_{i,N}$\ be the eigenvalues of\ $\sigma_{i}$, half the $\lambda_{i,j}%
$'s are $\frac{1/2+3\varepsilon}{N/2}$\ and the other half are $\frac
{1/2-3\varepsilon}{N/2}$; this can be seen by applying a unitary
transformation that diagonalizes $\sigma_{i}$, by rotating to a basis that
contains $\frac{N}{2}$\ basis vectors for $S_{i}$. \ Hence%
\begin{align*}
\operatorname{S}\left(  \sigma_{i}\right)    & =\sum_{x=1}^{N}\lambda
_{i,x}\log_{2}\frac{1}{\lambda_{i,x}}\\
& =\frac{N}{2}\left(  \frac{1/2+3\varepsilon}{N/2}\right)  \log_{2}\left(
\frac{N/2}{1/2+3\varepsilon}\right)  +\frac{N}{2}\left(  \frac
{1/2-3\varepsilon}{N/2}\right)  \log_{2}\left(  \frac{N/2}{1/2-3\varepsilon
}\right)  \\
& =\log_{2}N-\left[  1-\left(  \frac{1}{2}+3\varepsilon\right)  \log
_{2}\left(  \frac{1}{1/2+3\varepsilon}\right)  -\left(  \frac{1}%
{2}-3\varepsilon\right)  \log_{2}\left(  \frac{1}{1/2-3\varepsilon}\right)
\right]  \\
& \geq\log_{2}N-O\left(  \varepsilon^{2}\right)  .
\end{align*}
Combining,%
\[
\operatorname{I}\left(  \zeta;i\right)  =O\left(  T\varepsilon^{2}\right)  .
\]
We conclude that, in order to achieve quantum mutual information
$\operatorname{I}\left(  \zeta;i\right)  =\Omega\left(  N^{2}\right)  $, so
that $\zeta$\ can determine $i$,%
\[
T=\Omega\left(  \frac{N^{2}}{\varepsilon^{2}}\right)  =\Omega\left(
\frac{\min\left\{  D^{2},\log M\right\}  }{\varepsilon^{2}}\right)
\]
copies of $\sigma_{i}$\ are information-theoretically necessary.
\end{proof}

Let's also observe a case in which Theorem \ref{lowerq} has a matching upper bound.

\begin{proposition}
\label{gap}Given an unknown mixed state $\rho$, and known two-outcome
measurements $E_{1},\ldots,E_{M}$ and reals $c_{1},\ldots,c_{M}\in\left[
0,1\right]  $, suppose we're promised that for each $i\in\left[  M\right]  $,
either $\operatorname{Tr}\left(  E_{i}\rho\right)  \geq c_{i}$ or
$\operatorname{Tr}\left(  E_{i}\rho\right)  \leq c_{i}-\varepsilon$. \ Then we
can decide which is the case for each $i\in\left[  M\right]  $ using
$k=O\left(  \frac{\log M/\delta}{\varepsilon^{2}}\right)  $\ copies of $\rho$,
with success probability at least $1-\delta$.
\end{proposition}

\begin{proof}
For each $i\in\left[  M\right]  $, we simply perform a collective measurement
on $\rho^{\otimes k}$, which applies $E_{i}$\ to each copy of $\rho$, and
accepts if and only if the number of accepting invocations is at least
$\left(  c_{i}-\frac{\varepsilon}{2}\right)  k$. \ By a Chernoff bound, this
causes us to learn the truth for that $i$ with probability at least%
\[
1-\exp\left(  -\varepsilon^{2}\frac{\log M/\delta}{\varepsilon^{2}}\right)
\geq1-\frac{\delta^{2}}{4M^{2}},
\]
provided the constant in the big-$O$ is sufficiently large. \ By Lemma
\ref{qub}\ (the Quantum Union Bound), this means that by applying these
collective measurements successively, we can learn the truth for
\textit{every}\ $i\in\left[  M\right]  $\ with probability at least%
\[
1-2M\sqrt{\frac{\delta^{2}}{4M^{2}}}=1-\delta.
\]

\end{proof}

Of course, we also know from the recent work of O'Donnell and Wright
\cite{owright}\ and Haah et al.\ \cite{hhjwy} that $O\left(  \frac{D^{2}%
}{\varepsilon^{2}}\right)  $\ copies of $\rho$\ suffice to learn $\rho$\ up to
$\varepsilon$\ error in trace distance, so that number suffices for shadow
tomography as well.

\section{Open Problems\label{OPEN}}

This paper initiated the study of shadow tomography of quantum states, and
proved a surprising upper bound on the number of copies of a state that
suffice for it. \ But this is just the beginning of what one can ask about the
problem. \ Here we discuss four directions for future work.\bigskip

\textbf{(1) Tight Bounds.} \ What is the true sample complexity of shadow
tomography? \ We conjecture that Theorem \ref{main} is far from tight.\ \ Our
best current lower bound, Theorem \ref{lowerq}, is $\Omega\left(  \frac
{\min\left\{  D^{2},\log M\right\}  }{\varepsilon^{2}}\right)  $. \ Could
shadow tomography be possible with only $\left(  \frac{\log M}{\varepsilon
}\right)  ^{O\left(  1\right)  }$\ copies of $\rho$, independent of $D$---as
it is in the classical case (by Proposition \ref{classical}), and the case of
a promise gap (by Proposition \ref{gap})? \ Can we prove any lower bound of
the form $\omega\left(  \log M\right)  $?

Also, what happens if we consider measurements with $K>2$ outcomes? \ In that
setting, it's easy to give \textit{some} upper bound for the state complexity
of shadow tomography, by reducing to the two-outcome case. \ Can we do better?

One can also study the state complexity of many other learning tasks. \ For
example, what about what we called the \textquotedblleft gentle search
problem\textquotedblright: finding an $i\in\left[  M\right]  $\ for which
$\operatorname{Tr}\left(  E_{i}\rho\right)  $\ is large, promised that such an
$i$ exists? \ Can we improve our upper bound of $\widetilde{O}\left(
\frac{\log^{4}M}{\varepsilon^{2}}\right)  $\ copies of $\rho$\ for that task?
\ Or what about approximating the vector of $\operatorname{Tr}\left(
E_{i}\rho\right)  $'s in other norms, besides the $\infty$-norm?\bigskip

\textbf{(2) Shadow Tomography with Restricted Kinds of Measurements.} \ From
an experimental standpoint, there are at least three drawbacks of our shadow
tomography procedure. \ First, the procedure requires collective measurements
on roughly $\frac{\log D}{\varepsilon^{2}}$\ copies of $\rho$, rather than
measurements on each\ copy separately (or at any rate, collective measurements
on a smaller number of copies, like $\log\log D$). \ Second, the procedure
requires so-called \textit{non-demolition measurements}, which carefully
maintain a state across a large number of sequential measurements (or
alternatively, but just as inconveniently for experiments, an extremely long
circuit to implement a single measurement). \ Third, the actual measurements
performed are not just amplified $E_{i}$'s, but the more complicated
measurements required by Harrow et al.\ \cite{hlm}.

It would be interesting to address these drawbacks, either alone or in
combination. \ To illustrate, if one could prove the soundness of Aaronson's
original procedure for the Quantum OR Bound \cite{aar:qmaqpoly}, that would
remove the third drawback, though not the other two.\bigskip

\textbf{(3) Computational Efficiency.} \ To what extent can our results be
made \textit{computationally} efficient, rather than merely efficient in the
number of copies of $\rho$? \ In Section \ref{RELATED}, we estimated the
computational complexity of our shadow tomography procedure as $\widetilde{O}%
\left(  ML\right)  +D^{O\left(  \log\log D\right)  }$ (ignoring the dependence
on $\varepsilon$\ and $\delta$), where $L$ is the length of a circuit to
implement a single $E_{i}$. \ We also discussed the recent work of Brand\~{a}o
et al.\ \cite{bkllsw}, which builds on this work to do shadow tomography in
$\widetilde{O}\left(  \sqrt{M}L\right)  +D^{O\left(  1\right)  }$ time, again
using $\operatorname*{poly}\left(  \log M,\log D\right)  $\ copies of $\rho
$---or in $\widetilde{O}\left(  \sqrt{M}\operatorname*{polylog}D\right)
$\ time under strong additional assumptions about the $E_{i}$'s.

How far can these bounds be improved, with or without extra assumptions on
$\rho$\ or the $E_{i}$'s?

There are some obvious limits. \ If the measurements $E_{1},\ldots,E_{M}$\ are
given by $D\times D$\ Hermitian matrices (or by circuits of size $D^{2}$), and
if the algorithm first needs to load descriptions of all the measurements into
memory, then that already takes $\Omega\left(  MD^{2}\right)  $ time, or
$\Omega\left(  \sqrt{M}\right)  $\ time just to Grover search among the
measurements. \ Applying a measurement with circuit complexity $D^{2}$\ takes
$D^{2}$\ time. \ Outputting estimates for each $\operatorname{Tr}\left(
E_{i}\rho\right)  $ takes $\Omega\left(  M\right)  $\ time---although
Brand\~{a}o et al.\ \cite{bkllsw}\ evade that bound by letting their output
take the form of a quantum circuit to prepare a state $\sigma$\ such that
$\left\vert \operatorname{Tr}\left(  E_{i}\sigma\right)  -\operatorname{Tr}%
\left(  E_{i}\rho\right)  \right\vert \leq\varepsilon$\ for all $i\in\left[
M\right]  $, which seems fair.

If we hope to do even better than that, we could demand that the measurements
$E_{i}$\ be implementable by a \textit{uniform} quantum algorithm, which takes
$i$ as input and runs in time $\operatorname*{polylog}D$. \ Let's call a
shadow tomography procedure \textit{hyperefficient} if, given as input such a
uniform quantum algorithm $A$, as well as $k=\operatorname*{poly}\left(  \log
M,\log D,\frac{1}{\varepsilon}\right)  $\ copies of $\rho$, the procedure uses
$\operatorname*{poly}\left(  \log M,\log D,\frac{1}{\varepsilon}\right)
$\ time to output a quantum circuit $C$ such that $\left\vert C\left(
i\right)  -\operatorname{Tr}\left(  E_{i}\rho\right)  \right\vert
\leq\varepsilon$\ for all $i\in\left[  M\right]  $. \ Note that, in the
special case that $\rho$\ and the $E_{i}$'s are classical, hyperefficient
shadow tomography is actually possible, since we can simply output a circuit
$C$ that hardwires $k$ classical samples $x_{1},\ldots,x_{k}\sim\mathcal{D}$,
and then on input $i$, returns the empirical mean of $E_{i}\left(
x_{1}\right)  ,\ldots,E_{i}\left(  x_{k}\right)  $.

In the general case, by contrast, we observe the following:

\begin{proposition}
\label{hyper1}Suppose there exists a hyperefficient shadow tomography
procedure. \ Then$\ $quantum advice can always be simulated by classical
advice---i.e., $\mathsf{BQP/qpoly}=\mathsf{BQP/poly}$.
\end{proposition}

\begin{proof}
This already follows from the assumption that a quantum circuit $C$ of size
$\operatorname*{poly}\left(  \log M,\log D\right)  $, satisfying (say)
$\left\vert C\left(  i\right)  -\operatorname{Tr}\left(  E_{i}\rho\right)
\right\vert \leq\frac{1}{10}$ for every $i\in\left[  M\right]  $,
\textit{exists}. \ For a description of that $C$ can be provided as the
$\mathsf{BQP/poly}$\ advice when simulating $\mathsf{BQP/qpoly}$.
\end{proof}

Note that Aaronson and Kuperberg \cite{ak} gave a quantum oracle relative to
which $\mathsf{BQP/qpoly}\neq\mathsf{BQP/poly}$. \ By combining that with
Proposition \ref{hyper1}, we immediately obtain a quantum oracle relative to
which hyperefficient shadow tomography is impossible.

One further observation:

\begin{proposition}
\label{hyper2}Suppose there exists a hyperefficient shadow tomography
procedure. \ Then quantum copy-protected software (see \cite{aar:qcopy} or
Section \ref{MOTIV}) is impossible.
\end{proposition}

\begin{proof}
Given $n^{O\left(  1\right)  }$\ copies of a piece $\rho$\ of quantum
software, by assumption we could efficiently produce an $n^{O\left(  1\right)
}$-bit classical string $s$, which could be freely copied and would let a user
efficiently compute $\rho$'s output behavior (i.e., accepting or rejecting) on
any input $x\in\left\{  0,1\right\}  ^{n}$\ of the user's choice.
\end{proof}

It would be interesting to know what further improvements are possible to the
computational complexity of shadow tomography, consistent with the obstacles
mentioned above. \ Also, even if shadow tomography inherently requires
exponential computation time in the worst case, one naturally seeks do better
in special cases. \ For example, what if the $E_{i}$'s are stabilizer
measurements? \ Or if $\rho$\ is a low-dimensional matrix product
state?\bigskip

\textbf{(4) Applications.} \ A final, open-ended problem is to \textit{find
more applications} of shadow tomography. \ In Section \ref{MOTIV}, we gave
implications for quantum money, quantum software, quantum advice, and quantum
communication protocols. \ But something this basic being possible seems like
it ought to have further applications in quantum information theory, and
conceivably even experiment. \ In the search for such applications, we're
looking for any situation where (i) one has an unknown entangled state $\rho
$\ on many particles; (ii) one is limited mainly in how many copies of $\rho
$\ one can produce; (iii) one wants to know, at least implicitly, the
approximate expectation values of $\rho$ on a huge number of observables; and
(iv) one doesn't need to know more than that.

\section{Acknowledgments}

I thank Fernando Brand\~{a}o, Patrick Hayden, Cedric Lin, Guy Rothblum, and
Ronald de Wolf for helpful discussions and comments, Steve Flammia for
suggesting the name \textquotedblleft shadow tomography,\textquotedblright%
\ and the anonymous reviewers for their suggestions.

\bibliographystyle{plain}
\bibliography{thesis}

\end{document}